%% file: main.tex
\documentclass[journal]{IEEEtran}
\ifCLASSINFOpdf
\else
\fi
\pdfoutput=1
\usepackage{booktabs}       
\usepackage{amsfonts}       
\usepackage{xcolor}
\usepackage[linesnumbered,titlenumbered,ruled,vlined,commentsnumbered,noend]{algorithm2e}
\usepackage{algpseudocode}
\usepackage{graphicx}
\usepackage[tight]{subfigure}
\usepackage{grffile}

\usepackage{amssymb}
\usepackage{amsmath}
\usepackage{amsthm}
\usepackage{wrapfig}
\usepackage{authblk}
\usepackage{capt-of}
\DeclareMathOperator*{\argmin}{argmin}
\newtheorem{definition}{Definition}

\newtheorem{theorem}{Theorem}

\newcommand{\AlgName}{\textsc{Top-DP}\xspace}
\newcommand{\bheading}[1]{{\vspace{2pt}\noindent{\textbf{#1}}\hspace{2pt}}}
\newcommand{\rone}[1]{\textcolor{black}{#1}}
\newcommand{\rtwo}[1]{\textcolor{black}{#1}}
\newcommand{\ronetwo}[1]{\textcolor{black}{#1}}

\begin{document}
\title{Topology-aware Differential Privacy for Decentralized Image Classification}
\author{Shangwei Guo, Tianwei Zhang, Guowen Xu, Han Yu, Tao Xiang, and Yang Liu
    \thanks{T. Zhang is the corresponding author.}
    \thanks{S. Guo and T. Xiang are with College of Computer Science, Chongqing University, Chongqing, China (email: \{swguo, txiang\}@cqu.edu.cn).}
    \thanks{T. Zhang, G. Xu, H. Yu, and Y. Liu are with School of Computer Science and Engineering, Nanyang Technological University, Singapore (email: \{tianwei.zhang, guowen.xu, han.yu, and yangliu\}@ntu.edu.sg).}}

\maketitle

\begin{abstract}
Image classification is a fundamental artificial intelligence task that labels images into one of some predefined classes.  However, training complex image classification models requires a large amount of computation resources and data in order to reach state-of-the-art performance. This demand drives the growth of distributed deep learning, where multiple agents cooperatively train global models with their individual datasets. Among such learning systems, decentralized learning is particularly attractive, as it can improve the efficiency and fault tolerance by eliminating the centralized parameter server, which could be the single point of failure or performance bottleneck.

Although the agents do not need to disclose their training image samples, they exchange parameters with each other at each iteration, which can put them at the risk of data privacy leakage. Past works demonstrated the possibility of recovering training images from the exchanged parameters. One common defense direction is to adopt Differential Privacy (DP) to secure the optimization algorithms such as Stochastic Gradient Descent (SGD). Those DP-based methods mainly focus on standalone systems, or centralized distributed learning. How to enforce and optimize DP protection in decentralized learning systems is unknown and challenging, due to their complex communication topologies and distinct learning characteristics.


In this paper, we design \AlgName, a novel solution to optimize the differential privacy protection of decentralized image classification systems. The key insight of our solution is to leverage the unique features of decentralized communication topologies to reduce the noise scale and improve the model usability. (1) We enhance the DP-SGD algorithm with this \emph{topology-aware} noise reduction strategy, and integrate the time-aware noise decay technique. (2) We design two novel learning protocols (synchronous and asynchronous) to protect systems with different network connectivities and topologies. We formally analyze and prove the DP requirement of our proposed solutions. Experimental evaluations demonstrate that our solution achieves a better trade-off between usability and privacy than prior works. To the best of our knowledge, this is the first DP optimization work from the perspective of network topologies.

\end{abstract}
\IEEEpeerreviewmaketitle
\input{body/introduction}
\input{body/literature}
\input{body/problemdef}
\input{body/algorithm}

\input{body/experiments}
\input{body/conclusion}

\bibliographystyle{body/IEEEtran}
\bibliography{body/ref}
\end{document}

%% file: body/introduction.tex
\section{Introduction}\label{sec:introduction}
\IEEEPARstart{D}{eep} Learning (DL) has become one of the most popular and powerful machine learning methods for the image classification task.
To learn an accurate DL model, a common technique is Stochastic Gradient Descent (SGD), which iteratively approaches the ideal model by minimizing the empirical performance on a large number of training images. To accelerate the training process and protect data privacy, this SGD task can be distributed to multiple agents with their own training image sets to collaboratively learn a shared image classification model. This distributed learning \cite{nedic2009distributed,konevcny2016federated,abou2012fusion} has gained a lot of popularity, especially in the edge computing \cite{zhao2019dynamic,yang2018high,chen2015robust}.

Distributed learning can be divided into two categories: centralized and decentralized learning \cite{lian2017can}. A centralized learning system utilizes a centralized parameter server to collect and aggregate estimates (i.e., model parameters) of agents at each iteration. In a decentralized system, agents interconnect based on a certain network topology and exchange estimates with their neighbors to reach consensus on the DL model. Distributed learning can prevent direct privacy leakage as each agent keeps its own private dataset locally at the training stage. However, it still faces the threats of indirect privacy leakage: the exchanged estimates among agents may contain information about their training sets. This gives honest-but-curious agents opportunities to compromise the data privacy of their neighbors.
Past works have demonstrated the feasibility and severity of model inversion attacks \cite{hitaj2017deep,zhu2019deep,he2019model,zhu2020private} and membership inference attacks \cite{melis2019exploiting,leino2020stolen} in distributed learning.


To mitigate such privacy threats in distributed training, one promising solution is Differential Privacy (DP), which was originally introduced to preserve the privacy of individual data records in statistical databases \cite{dwork2006our}.
A number of studies have then applied DP to SGD to enhance the privacy of DL training in different environments \cite{chaudhuri2011differentially,abadi2016deep,zhang2018improving,li2018differentially,yu2019differentially,jayaraman19evaluatiing}. Most existing DP-SGD algorithms adopt additive noise mechanisms by adding random noise to the estimates in every training iteration. There exists a trade-off between privacy and usability, determined by the noise scale added during training: adding too much noise can meet the privacy requirements, at the cost of huge drop in model accuracy. As a result, it is critical to identify the minimal amount of noise that can provide desired privacy protection, and also maintain acceptable model performance.

Two common approaches were devised to optimize the DP mechanism and balance the privacy-usability trade-off. The first one is to carefully restrict the sensitivity of randomized mechanisms. For example, Abadi et al. \cite{abadi2016deep} bounded the influence of training samples on gradients by clipping each gradient in $l_2$ norm below a given threshold. Yu et al. \cite{yu2019differentially} optimized the model accuracy by adding decay noise to the gradients over the training time since the learned models converge iteratively. The second approach is to precisely track the accumulated privacy cost of the training process using composition techniques such as the strong composition theorem \cite{dwork2010boosting} and moments account (MA) \cite{abadi2016deep, bhowmick2018protection,hynes2018efficient,kang2019weighted,gong2020preserving}.
%

Those DP-SGD solutions have been well developed and evaluated in centralized learning systems. In contrast, privacy protection in the decentralized learning setting is less explored. There are distinct differences between these two systems. First, decentralized systems have more interactions and parameter exchanges in order to reach the consensus. Each agent receives parameters from multiple neighbors and broadcasts the update to them.
Second, many decentralized systems are usually spontaneously organized and each agent is relatively independent. It is highly possible that certain nodes are offline due to the discrepancy of network bandwidth or unpredictable system faults. The asynchronous training mode \cite{lian2018asynchronous,luo2019heterogeneity} thus becomes more prevalent with higher reliability and efficiency. Hence, we raise two questions: (1) \emph{how can we design DP algorithms to support these unique features (e.g., \rtwo{decentralized topology}, asynchronous training)?} (2) \emph{How can we leverage these features to further optimize the DP solution and balance the privacy-usability trade-off?}

Prior works in differentially private decentralized learning mainly targeted the Alternating Direction Method of Multipliers (ADMM) algorithm with existing optimization techniques \cite{zhang2016dynamic,zhang2018improving,li2018differentially,ding2019optimal}. They cannot be used with the mainstream SGD-based training tasks. \rtwo{Other differentially private decentralized leanring methods \cite{li2018differentially,lou2017privacy,lu2020privacy,zhou2019privacy,hou2019differential} are designed either using existing DP techniques or for a specific application.  For example, \cite{li2018differentially} simply applied the standard DP technique (e.g., tracking accumulated privacy loss \cite{dwork2010boosting}) from the centralized setting to the decentralized one.} These optmization techniques have been well studied, and seem to reach the performance limit. In contrast, the unique topology features were never considered.

In this paper, we present \AlgName, a novel \textbf{Top}ology-aware \textbf{D}ifferential \textbf{P}rivacy approach for SGD-based training in decentralized systems.
\AlgName leverages network features of decentralized systems to optimize the randomized mechanism. The key idea is that each agent takes into account the injected noise from its neighbors when adding its own noise to the aggregated parameter. Such noise reuse can significantly reduce the actual noise scale added by each agent, but still satisfying the DP requirement. In addition, \AlgName can also be integrated with the noise decay technique from the standalone training mode, to further optimize the DP protection in decentralized systems.

Based on this strategy, we design two new learning protocols to realize our optimization. The first one is for synchronous training mode. Different from existing styles, each agent calculates and sends different aggregated estimates to different neighbors. This can guarantee that each parameter exchange can always enjoy the maximal benefit from the topology-aware strategy. However, such advantage becomes minor when two connected agents share the same neighbors in a network topology. To copy with this corner case and save communication bandwidth, we introduce an asynchronous training protocol: at every iteration, each agent only pairs with one neighbor which is randomly picked to meet the noise reduction criterion. Then the parameter exchange between the pair can reduce the noise scale, and eliminate unnecessary communication costs in total.

We extensively validate the privacy and effectiveness of our proposed solution. From the theoretical view, we formally prove that each agent can guarantee differential privacy with significantly reduced noise. Empirically, we conduct comprehensive experiments to demonstrate that our solution outperforms prior works and techniques under various system configurations, datasets and DL models. We make the following contributions in this paper:

\begin{itemize}
    \item To best of our knowledge, this is the first work that utilizes the network topology feature to enhance the usability of DP in distributed learning systems.
    \item We propose two novel learning protocols to achieve DP optimization for both synchronous and asynchronous training modes.
    \item We formally prove that our solution can guarantee the DP requirement for all the agents, and analyze its advantage under different decentralized settings.
    \item We conduct extensive experiments to show the superior of our method over prior works with various scenarios and image classification tasks.
\end{itemize}

The rest of this paper is organized as follows. Section~\ref{sec:problemdef} introduces formal definitions of decentralized systems, differential privacy and problem statement. Section~\ref{sec:algorithm} presents the topology-aware and time-aware strategies for decentralized systems. Section~\ref{sec:protocol} illustrates two learning protocols for synchronous and asynchronous training modes, respectively. Section \ref{sec:analysis} presents our privacy analysis, and Section~\ref{sec:experiments} shows the experimental evaluations of our approach under various system settings. We review the related works in Section~\ref{sec:literature}, and conclude in Section \ref{sec:conclusion}.

%% file: body/literature.tex
\section{Related Work}\label{sec:literature}
Differential privacy has been adopted to protect the individual privacy of training datasets and a large amount of DP-SGD algorithms have been proposed \cite{jayaraman19evaluatiing,song2021systematic}. We classify these algorithms into two categories: DP-SGD for standalone and distributed learning systems.
\subsection{DP-SGD for Standalone Learning Systems}
For standalone systems, there are commonly two possible ways to add random noise. The first one is to inject noise to the objective function. For instance, Chaudhuri et al. \cite{chaudhuri2011differentially} perturbed the objective
function before optimizing over classifiers and proved that the objective perturbation is DP if certain convexity and differentiability criteria hold. The sensitivity analysis methods in \cite{chaudhuri2011differentially} relies on a strong convexity assumption. However, most objective function is non-convex. Phan et al. \cite{phan2016differential} attempted to use the objective perturbation by replacing the non-convex function with a convex polynomial function. To this end, a new convex polynomial function was introduced in \cite{phan2016differential} to approximate the non-convex one. However, this would change the learning protocol and, even worse, sacrifice the model's performance.

A simpler but more popular way is to add random noise to the gradients. Abadi et al. \cite{abadi2016deep} achieved DP by adding Gaussian noise to the gradients of each iteration. This approach restricts the sensitivity of randomized mechanisms, i.e., the influence of training data on gradients, by clipping each gradient in $l_2$ norm below a given threshold. Abadi et al. \cite{abadi2016deep} also proposed MA to reduce the added noise by keeping track of a bound on the moments of the privacy loss during the training process. Yu et al. \cite{yu2019differentially}  focused on the DP problem during the sharing and publishing of pre-trained models. They optimized the model accuracy by adding decay noise to the gradients over the training time since the learned models converge iteratively. They improved the model usability by employing a generalization of concentrated DP, based on the observation that the privacy loss of an additive noise mechanism follows a sub-Gaussian distribution.

Another way to improve the model usability lies in precisely tracking the overall privacy cost of the training process. Shokri et al. \cite{shokri2015privacy} and Wei et al. \cite{wei2020federated} composed the additive noise mechanisms using the advanced composition theorem \cite{dwork2010boosting}, leading to a linear increase in the privacy budget. In \cite{abadi2016deep, bhowmick2018protection,hynes2018efficient,kang2019weighted}, moments account (MA) was used to reduce the added noise by keeping track of a bound on the moments of the privacy loss during the training process.
Other algorithms \cite{park2017dp,jayaraman2018distributed,yu2019differentially} were designed to improve the model usability using (zero) concentrated DP \cite{dwork2016concentrated}, based on the observation that the privacy loss of an additive noise mechanism follows a sub-Gaussian distribution. Recently, Asoode et al \cite{asoodeh2021three} proposed an optimal DP analysis to further reduce the scale of added noise during the training process.

\subsection{DP-SGD for Centralized Learning Systems}
Some works~\cite{shokri2015privacy,bhowmick2018protection,hynes2018efficient,jayaraman2018distributed,kang2019weighted} applied the DP techniques from the standalone mode to the centralized learning systems to preserve the privacy of the training data for each agent. For example, Shokri et al. \cite{shokri2015privacy} proposed a privacy-preserving distributed learning algorithm by adding Laplacian noise to each agent's gradients to prevent indirect leakage. Kang et al. \cite{kang2019weighted} adopted weighted aggregation instead of simply averaging to reduce the negative impact caused by uneven data scale.

In terms of the accumulated privacy loss, Kang et al. \cite{kang2019weighted} employed MA to track the overall privacy cost of the training process. Wei et al. \cite{wei2020federated} perturbed agents' trained parameters locally by adding Gaussian noise before uploading them to the server for aggregation and bounded the sensitivity of the Gaussian mechanism by clipping. Shokri et al. \cite{shokri2015privacy} and Wei et al. \cite{wei2020federated} composed the additive noise mechanisms using the strong composition theorem \cite{dwork2010boosting}, leading to a linear increase in the privacy budget.

\rtwo{Federated learning, as a typical example of centralized learning systems, has gained great popularity. A variety of works \cite{geyer2017differentially,truex2019hybrid,choudhury2019differential,wei2020federated,wei2021user} have attempted to solve the privacy problem using the above DP techniques in such federated systems. For instance, Truex et al. \cite{truex2019hybrid} proposed a hybrid method that leverages both secure multiparty computation and differential privacy techniques to achieve privacy-preserving federated learning. Choudhury et al. \cite{choudhury2019differential} designed a federated learning system that enables to process sensitive health data with differential privacy protection. Wei et al. \cite{wei2021user} provided user-level privacy protection for federated learning systems and improve the usability of trained models.}

\rtwo{Some DP-SGD methods \cite{wei2021user} for centralized learning systems, especially for federated learning systems, can be applied to decentralized systems as well. However, these methods are designed specifically for the corresponding centralized collaboration schemes, and do not well optimized for the decentralized setting. In contrast, our \AlgName utilizes the unique features of decentralized systems and can significantly improve the usability of the trained models of all agents compared with state-of-the-art DP-SGD for federated learning.}

\subsection{DP-SGD for Decentralized Learning Systems}
\rtwo{Several DP approaches \cite{zhang2016dynamic,zhang2018improving,li2018differentially,ding2019optimal,zhou2019privacy,lu2020privacy} were proposed for decentralized learning systems. However, they are mainly for the ADMM or gradient tracking algorithms, while there are very few solutions for the SGD algorithm in decentralized systems. Recently, motivated by the privacy leakage problem in big data analytics, Li et al. \cite{li2018differentially} proposed DP-SGD algorithms by adding Gaussian and Laplacian noise to the gradients. Zhou et al. \cite{zhou2019privacy} designed a differential privacy decentralized learning system for social recommendation systems. However, they track the accumulated privacy loss using the existing DP techniques such as the strong composition theorem, which limits the performance of learned models.}

Similar as the centralized DP approaches, all those decentralized DP solutions (both for ADMM, SGD or other optimization algorithms) only apply existing DP techniques and focus on restricting the sensitivity of the optimization algorithm. Besides, they require the decentralized systems to be well synchronous. In contrast, our \AlgName provides a novel optimization direction from the network topology. It can also be combined with other optimization algorithms such as ADMM and existing techniques (e.g., noise decay). The learning protocols in \AlgName can be applied to different training modes efficiently.

%% file: body/problemdef.tex
\section{Background and Problem Statement}\label{sec:problemdef}
In this section, we first formalize decentralized systems. Then, we present the threat model and the definition of DP for decentralized learning.
\subsection{Decentralized Systems}
We consider a decentralized system whose communication topology can be represented as an undirected graph: $\mathcal{G} = (V, E)$. $V$ denotes a set of participates (or agents) in this decentralized network. $E$ represents the set of communication links among the agents, with the following two properties:
\begin{enumerate}
\item  $(i, j) \in E$ if and only if agent $i$ can receive information from agent $j$;
\item $(j, i) \in E$ if $(i, j) \in E$.
\end{enumerate}
We assume this undirected graph is fully connected, i.e., giving two arbitrary agents $i$ and $j$, there always exists at least one path that connects them. This property can guarantee that information can be exchanged among all agents \cite{tsitsiklis1984problems,nedic2009distributed}.

\begin{figure}[t]
	\centering
	\includegraphics[width=\columnwidth]{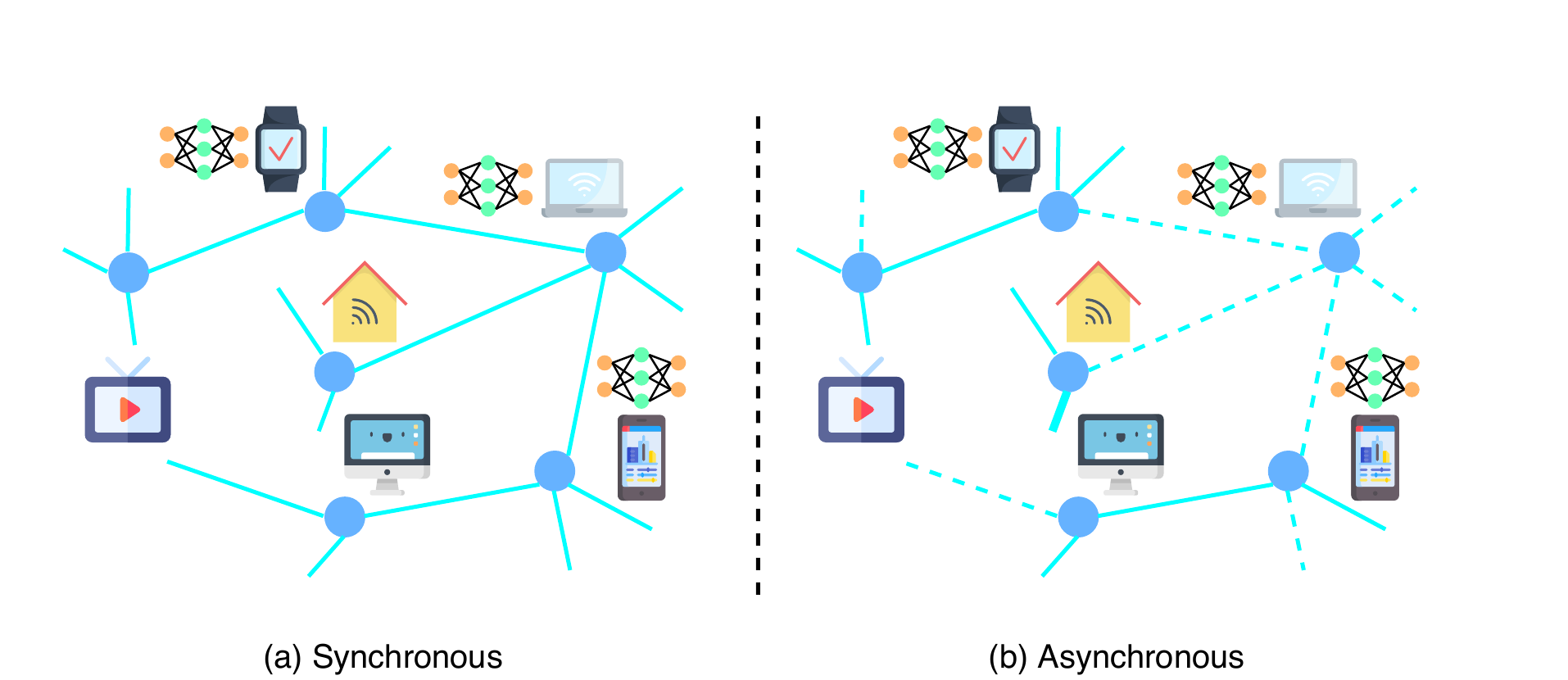}\label{fig:framework_decen}
	\caption{(a) Synchronous training and (b) asynchronous training in a decentralized learning system. Solid lines represent network connections with parameter exchanges between agents; dotted lines represent connections not used for parameter exchanges in certain iteration.}
	\label{fig:framework}
\end{figure}

\rone{Let $x \in \mathbb{R}^d$ be the $d$-dimensional estimate vector of a DL model. Each agent $i \in V$ obtains a private training dataset $D_i$, consisting of independent and identically distributed (i.i.d.) data samples from a distribution $D$. Those agents train a shared model by solving the optimization problem \cite{lian2017can,nedic2009distributed}: $$\min_{x\in \mathbb{R}^d} \mathbb{E}_{\xi\sim D}l(x;\xi),$$ where $\xi$ is a training data sample from $D$.
During training, each agent $i$ calculates its local estimate $x_{i}$, and exchanges $x_{i}$ with its neighbors for parameter update. There are basically two training modes for this iterative process. In the \emph{synchronous} mode (Fig. \ref{fig:framework} (a)), each agent $i$ needs to receive the estimates from all its neighbors before updating the model. In the \emph{asynchronous} mode (Fig. \ref{fig:framework} (b)), agent $i$ exchanges parameters with only part of its neighbors for model update. This happens when the agent just wants to choose a smaller number of neighbors for lower communication and computation cost, or when some of its neighbors fail to respond due to unexpected system or network faults.}

To adapt to both synchronous and asynchronous modes as well as maintaining the convergence rate, at each iteration, agent $i$ (1) first collects estimates from its neighbor(s); (2) randomly selects a neighbor $j^*$ from the participated neighbor(s); (3) utilizes the following update rule \cite{lian2018asynchronous,blanchard2017machine} to aggregate estimates and calculate the local estimate:
\begin{equation}
    x_{i} = \alpha x_{i} + (1-\alpha) x_{j^*}  - \lambda g(x_{i}, \xi_{i})
\end{equation}
where $\alpha \in [0,1]$ is a hyper-parameter determining the weight of the local estimate; $\lambda$ is the learning rate; $g(x_{i}, \xi_{i})$ is the stochastic gradient with $\xi_{i} \in D_i$. The gradient can also be replaced by a mini-batch of stochastic gradients \cite{lian2017can,lian2018asynchronous}.

\subsection{Threat Model}
\rone{In a decentralized system, we assume the agents are \emph{honest-but-curious}: all the agents agree on the proposed learning protocol and objective in advance. They will also strictly follow the steps of training and exchanging parameters during collaborative training. However, there exist some suspicious agents who attempt to passively steal the information and properties of their neighbors' datasets by analyzing the model parameters received at each iteration. We further assume that these suspicious agents will not collude to conduct the privacy attacks. Only connected neighbors are allowed to exchange information following the distributed training protocol.}

\rtwo{For decentralized learning systems, agents are connected directly or indirectly. Our goal is to adopt DP to protect the training data privacy of all agents.} \rone{DP is a rigorous mathematical framework to protect the privacy of individual records in a database when the aggregated information about this database is shared among untrusted parties \cite{dwork2006our}. } \rone{Thus, we formally define decentralized learning with DP as follows:}

\begin{definition}(DP of Decentralized Learning)
    A decentralized learning system is $\{(\epsilon_i, \delta_i)\}_{i \in V}$ differentially private if for each agent $i$, the randomized mechanism $\mathcal{M}_i: \mathcal{D}_i \rightarrow \mathcal{R}$ with domain $\mathcal{D}_i$ and range $\mathcal{R}$ satisfies $(\epsilon_i, \delta_i)$-DP, i.e., if for any two neighboring datasets $D_i, D_i'$ and any subset of outputs $S \subseteq  R$, the following property is held:
    \begin{equation}
        Pr[\mathcal{M}(D_i) \in S] \leq e^{\epsilon_i}Pr[\mathcal{M}(D_i') \in S] + \delta_i.
    \end{equation}
\end{definition}
$\mathcal{M}_i$ is restricted by two parameters: $\epsilon_i$ and $\delta_i$. $\epsilon_i$ is the privacy budget of agent $i$ to limit the privacy loss of training data. $\delta_i$ is a relaxation parameter that allows the privacy budget of $\mathcal{M}_i$ to exceed $\epsilon_i$ with probability $\delta_i$.
A decentralized learning system is differentially private if all agents are differentially private. Each agent can set its own privacy budget. Alternatively, the entire system can enforce a uniform privacy budget for all agents.

To achieve differentially private decentralized learning, a common and straightforward way is to use additive noise mechanisms at each iteration \cite{jayaraman19evaluatiing}. Specifically, we use Gaussian mechanism and denote $\sigma_i$ as the noise parameter of agent $i$. At each iteration, agent $i$ adds the Gaussian noise, $G_i = G(\sigma_{i}^2)$, to the updated local estimate to guarantee differential privacy (Eq. \ref{eq:aggre-dp}). Then, $i$ sends $\widetilde{x}_{i}$ to its neighbors.

\begin{equation}
\label{eq:aggre-dp}
    \widetilde{x}_{i} = \alpha \widetilde{x}_{i} + (1-\alpha) x_{j^*} - \lambda g(\widetilde{x}_{i}, \xi_{i}) + G_i.
\end{equation}

%% file: body/algorithm.tex
\section{Optimization Strategies}\label{sec:algorithm}


As shown in Eq. \ref{eq:aggre-dp}, the random noise $G_i$ added into the aggregated estimate must be large enough to satisfy the privacy requirement. However, adding too much noise can affect the model accuracy. So it is important to balance this trade-off.
This section presents the strategies adopted in \AlgName to reduce the amount of noise for each agent to improve the usability of trained models, without violating the DP requirement. We start with a novel topology-aware noise reduction strategy. Then we extend time-aware noise decay to decentralized systems.


\subsection{Strategy 1: Topology-aware Noise Reduction}
\label{sec:strategy1}

Existing DP-SGD solutions all assume that the required noise scale only depends on the agents themselves. In decentralized systems, the communication topology can affect the amount of noise as well. Our topology-aware noise reduction strategy is able to reduce the noise scale of each agent when considering its connectivity with its neighbors. The key insight of our approach is that \emph{the received estimates from other neighbors also contain certain noise, which can contribute to the noise scale of the aggregated estimate, thus reducing the amount of noise added by the agent itself.}

\begin{figure}{t}

    \includegraphics[width=0.9\columnwidth]{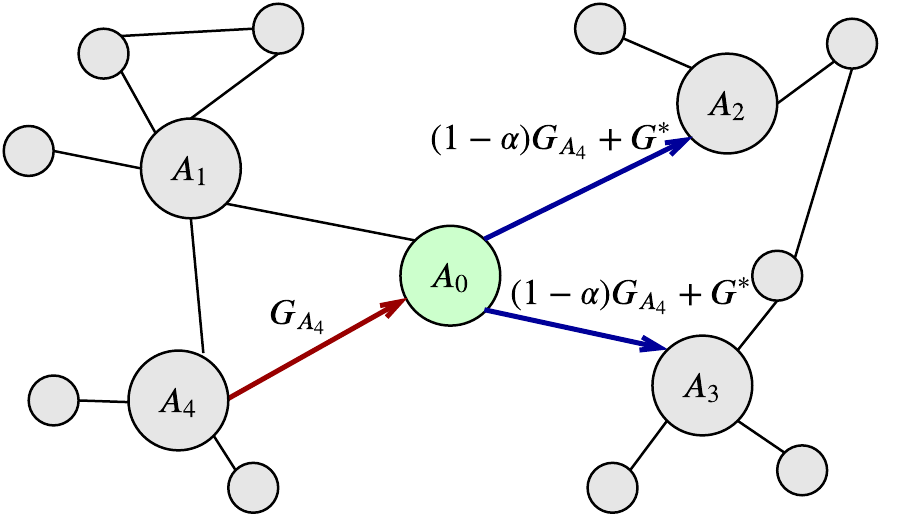}
    \caption{An illustrative example of topology-aware noise reduction.}
    \label{figs:connection}

\end{figure}

Fig. \ref{figs:connection} gives an illustrative example. We consider an agent $A_0$ with four neighbors, where $A_1$ and $A_4$ are connected as well. When $A_0$ obtains all estimates of its neighbors, we assume it picks the estimate $\widetilde{x}_{A_4}$ of $A_4$ for aggregation with its own estimate and gradient. Since the received $\widetilde{x}_{A_4}$ also includes Gaussian noise $G_{A_4}$, then the aggregated estimate following Eq. \ref{eq:aggre-dp} will have the corresponding random component $(1-\alpha)G_{A_4}$. As a result, when generating the estimate for agent $A_2$ or $A_3$, $A_0$ does not need to add the full-scale noise $G_{A_0}$. It only needs to inject the noise $G^*$ such that $$G_{A_0}=G^*+(1-\alpha)G_{A_4},$$ which can meet the DP requirement, but reduce the actual amount of noise.

It is worth noting that the noise scale $G^*$ is not applicable when generating estimates for $A_1$ or $A_4$. For $A_4$, since it already knows its own parameter $\widetilde{x}_{A_4}$, then $G_{A_4}$ is not random noise anymore. It is similar for $A_1$ as it receives $\widetilde{x}_{A_4}$ from $A_4$. Then for these two agents, we can pick another agent (e.g., $A_2$ or $A_3$) and generate a different estimate for them with still reduced noise scale. It is worth noting that our strategy allows the agent to send different estimates to different neighbors in one iteration, which is different from conventional distributed learning systems.

Formally, given an agent $i$, for each of its neighbors $j\in \mathcal{N}_i$, we define $$\mathcal{N}_i^j = \mathcal{N}_i \backslash (j\bigcup \mathcal{N}_j),$$ which is the set of $i$'s neighbors that are not connected to $j$ (or $j$ itself). For instance, in Fig. \ref{figs:connection}, we have $\mathcal{N}_{A_0}^{A_1}=\mathcal{N}_{A_0}^{A_4}=\{A_2, A_3\}$, $\mathcal{N}_{A_0}^{A_2}=\{A_1,A_3,A_4\}$ and $\mathcal{N}_{A_0}^{A_3}=\{A_1,A_2,A_4\}$. This also means that $j$ can be used in the aggregation for all agents in $\mathcal{N}_i^j$ with the reduced noise scale. Then our goal is to find a minimal set $\widehat{\mathcal{N}_i}$, such that using the agents inside this set for aggregation can cover all the neighbor agents of $i$. Note that there can exist a neighbor $j'$ that is connected to every neighbor in $\mathcal{N}_i$. Then we cannot find an non-adjacent neighbor to cover it, and should exclude it from $\mathcal{N}_i$. This process is described in Eq. \ref{eq:set-cover}. We will solve it heuristically in Section \ref{subsec:protocol}.

\begin{equation}
\label{eq:set-cover}
\widehat{\mathcal{N}_i} = \argmin_{\mathcal{N} \subseteq \mathcal{N}_i} (\bigcup_{j \in \mathcal{N}} \mathcal{N}_i^j = \mathcal{N}_i\backslash(\bigcup_{\mathcal{N}_i^{j'}=\emptyset} j'))
\end{equation}

After identifying $\widehat{\mathcal{N}_i}$, for $\forall j \in \mathcal{N}_i$, if $j$ is connected to every neighbor in $\mathcal{N}_i$ (i.e., $\mathcal{N}_i^{j}=\emptyset$), then agent $i$ just sends the local estimate with full-scale noise to $j$. Otherwise, there exists at least one neighbor $k \in \widehat{\mathcal{N}_i}$ such that $j\in \mathcal{N}_i^k$. Then the noise scale from $i$ to $j$, $G_i^j = G({\sigma_{i}^j}^2)$, should satisfy Eq. \ref{eq:aggre-red}(a) in order to guarantee the DP requirement against $j$, where $G_i, G_k$ are the full-scale noise. According to the additivity of Gaussian distribution, we calculate the noise parameter $\sigma_{i}^j$ via Eq. \ref{eq:aggre-red}(b). With this reduced noise scale, agent $i$ can update the estimate for agent $j$ based on Eq. \ref{eq:aggre-red}(c).

\begin{subequations}
\label{eq:aggre-red}
  \begin{align}
& G_i = (1-\alpha)G_k + G_i^j \\
& \sigma_{i}^j = \sqrt{\sigma_{i}^2 - (1- \alpha)^2\sigma_{k}^2} \\
& \widetilde{x}_{i}^j = \alpha\widetilde{x}_{i} + (1-\alpha)\widetilde{x}_{k}^i  - \lambda g(\widetilde{x}_{i}, \xi_s) + G_i^j
\end{align}
\end{subequations}

\subsection{Strategy 2: Time-aware Noise Decay}
\label{sec:nose-decay}
Our topology-aware strategy can be combined with existing state-of-the-art techniques from other systems to enhance the optimization effects. We use the time-aware noise decay as an example. This technique was originally proposed in \cite{yu2019differentially}, to optimize the DP protection of model training in standalone systems. Here we apply this technique to decentralized systems. The key idea is that \emph{the model converges and the norm of gradients decreases as the training iteration increases. Thus, the sensitivity of the Gaussian mechanism decreases, allowing us to inject less noise to the gradients}. Note that the training datasets are distributed in different agents, all agents in the decentralized system should reach a consensus on the noise decay schedule to tolerate the differences in the datasets.

Specifically, compared to the aggregation process in Eq. \ref{eq:aggre-red}(c), our first modification is to clip the gradients in $l_2$ norm to bound their size at each training iteration. We follow the method from \cite{abadi2016deep}: given a clipping threshold $C$, the clipped gradient vector $\bar{g}$ is bounded by $C$, as shown in Eq. \ref{eq:noise-decay}(a).

Our second modification is to dynamically reduce the noise scale over the training time. Without loss of generality, we use step decay to reduce the noise scale every few epochs. Let $\sigma_{0, i}$ be the initial noise parameter of agent $i$. The noise parameter of agent $i$ at the $t$-th iteration is shown in Eq. \ref{eq:noise-decay}(b), where $\gamma \in (0, 1)$ is the reduction factor and $period$ is the reduction step of noise decay.

\begin{subequations}
\label{eq:noise-decay}
  \begin{align}
& \bar{g} (\widetilde{x}_{i}, \xi_s) = \frac{g(\widetilde{x}_{i}, \xi_s)}{\max(1, \frac{|g(\widetilde{x}_{i}, \xi_s)|}{C})} \\
& \sigma_{t, i} = \text{Decay}(\sigma_{0, i}, t) = \sigma_{0, i}\gamma^{\lfloor \frac{t}{period} \rfloor }
\end{align}
\end{subequations}

\section{Topology-aware Learning Protocols}\label{sec:protocol}
With the topology-aware and time-aware strategies, we design two end-to-end decentralized learning protocols for synchronous and asynchronous modes, respectively.

\subsection{Synchronous Topology-aware Protocol}\label{subsec:protocol}
In the synchronous mode, each agent requires the estimates from all its neighbors at each iteration. Then it solves Eq. \ref{eq:set-cover} and calculates the parameters for different neighbors following Eqs. \ref{eq:aggre-red} and \ref{eq:noise-decay}. \rtwo{We design a synchronous topology-aware protocol to support DL on decentralized topologies and Algorithm \ref{alg:synchronous} illustrates the detailed communication and learning process of an agent $i$.}

The algorithm takes as input the initial estimate $x_0$, initial noise parameter $\sigma_{0, i}$, learning rate $\lambda$, and number of iterations $T$.
Before iteratively optimizing the shared model, agent $i$ sends $\sigma_{0, i}$ to and receives $\{\sigma_{0, j}\}_{j \in \mathcal{N}_i}$ from its neighbors (Line \ref{line:broadcast}). For $j \in \mathcal{N}_i$, agent $i$ computes the neighbor set $\mathcal{N}_i^j$, including all neighbors that do not connect with $j$.
Then, it updates the initial estimate and sends $\widetilde{x}_{0, i}$ to its neighbors (Lines \ref{line:initial}-\ref{line:initial_end}).
At the $i$-th iteration, it first computes the full-scale noise parameter $\sigma_{t, i}$ using the time-aware noise decay strategy (Line \ref{line:decay}). Then, it computes the clipped gradient using a randomly selected sample $\xi_s$, generates estimates for all its neighbors and updates its local estimate using the proposed topology-aware noise reduction strategy.

To heuristically solve Eq. \ref{eq:set-cover}, agent $i$ continuously selects an agent $k$ from $\mathcal{N}_{i}$ exclusively until all agents are traversed or $\widehat{\mathcal{N}_i}$ is found. For $\forall j \in \mathcal{N}_{i}^k$, it computes the estimate $\widetilde{x}_{i}^j$  and sends it to $j$ (Lines \ref{line:network}-\ref{line:network_end}). The complexity of the approximate method is $O(|\mathcal{N}_{i}|d)$. Then, agent $i$ randomly selects a neighbor $j^* \in \mathcal{N}_{i}$ and updates its local noised estimate  (Lines \ref{line:sample}-\ref{line:update}).
If there are still uncovered neighbors, $i$ sends its local estimate to those neighbors (Lines \ref{line:if}-\ref{line:if_end}). After $T$ iterations, Algorithm \ref{alg:synchronous} returns the final differentially private DL model.

\begin{algorithm}[t!]
    \SetAlgoLined
    \SetKwInOut{Input}{Input}
    \SetKwInOut{Output}{Output}
    \Input{Initial estimate $x_0$, initial budget $\sigma_{0,i}$, learning rate $\lambda$, number of iterations $T$}
    Send $\sigma_{0, i}$ to $\mathcal{N}_i$ and receive $\{\sigma_{0,j}\}_{j \in \mathcal{N}_i}$ \label{line:broadcast}\;
    \ForEach{$j \in \mathcal{N}_i$}{
        $\mathcal{N}_i^j \gets \mathcal{N}_i \backslash (j\bigcup \mathcal{N}_j)$
    }
    $\bar{g}(x_0, \xi_s)$ ~$\gets$~ Compute the clipped gradient\label{line:initial}\;
    $\widetilde{x}_{i} \gets x_0 - \lambda\bar{g}(x_0, \xi_s) + G(\sigma_{0,i}^2C^2)$\;
    Send $\widetilde{x}_{i}$ to its neighbors\label{line:initial_end}\;
    \For{$t \in [0, T)$}{
        $\sigma_{t, i} \gets \texttt{Decay}(\sigma_{0, i}, t)$ \label{line:decay}\;
        $\bar{g}(\widetilde{x}_{i}, \xi_s)$ ~$\gets$~ Compute the clipped gradient\;
        $\mathcal{N}_i^* \gets \mathcal{N}_i$ and $\mathcal{N}_i^{\ddagger} \gets \mathcal{N}_i$ \label{line:network}\;
        \While{$\mathcal{N}_i^* \neq \emptyset$ and $\mathcal{N}_i^{\ddagger} \neq \emptyset$}{
            Randomly select $k \in \mathcal{N}_i^{\ddagger}$ and $\mathcal{N}_i^{\ddagger} \gets \mathcal{N}_i^{\ddagger} \backslash k$\;
            $\sigma_{t, i}^j \gets \sqrt{\sigma_{t, i}^2 - (1- \alpha)^2\sigma_{t, k}^2}$, where $\sigma_{t, k} = \texttt{Decay}(\sigma_{0, k}, t)$\;
            \ForEach{$j \in \mathcal{N}_i^* \cap \mathcal{N}_i^{k}$}{
                Update the estimate $\widetilde{x}_{i}^j \gets \alpha\widetilde{x}_{i} + (1-\alpha)\widetilde{x}_{k}^i  - \lambda\bar{g}(\widetilde{x}_{i}, \xi_s) + G({\sigma_{t, i}^j}^2C^2)$\;
                Send $\widetilde{x}_{i}^j$ to agent $j$\;
            }
            $\mathcal{N}_i^* \gets \mathcal{N}_i^*\backslash \mathcal{N}_i^{k}$\label{line:network_end}\;
        }
        Randomly select an agent $j^*$ from $\mathcal{N}_i$ \label{line:sample}\;
        Update the local estimate $\widetilde{x}_{i} \gets \alpha\widetilde{x}_{i} + (1-\alpha)\widetilde{x}_{j^*}^i  - \lambda\bar{g}(\widetilde{x}_{i}, \xi_s) + G(\sigma_{t, i}^2C^2)$ \label{line:update}\;
        \If{$\mathcal{N}_i^* \neq \emptyset$ \label{line:if}}{
            \ForEach{$j \in \mathcal{N}_i^*$}{
                $\widetilde{x}_{i}^j \gets \widetilde{x}_{i}$ and send $\widetilde{x}_{i}^j$ to agent $j$ \label{line:if_end}\;
            }
        }
    }
    \Return $\widetilde{x}_{i}$
    \caption{Differentially private decentralized learning for agent $i$ in the synchronous mode.}\label{alg:synchronous}
\end{algorithm}

\subsection{Asynchronous Topology-aware Protocol}
Although the synchronous training in Algorithm \ref{alg:synchronous} can realize the proposed strategies to improve the model usability, it still leaves some spaces for further optimization. First, each agent only selects part of the received parameters for update while discarding the rest. So it is not necessary to collect the estimates from all the neighbors, which can cause extra communication cost and waiting latency. Second, as introduced in Section \ref{sec:strategy1}, when agent $j$ connects to every neighbor of agent $i$, $i$ has to add full-scale noise to the parameter sent to $j$. The topology-aware optmization will lose effectiveness when there are a lot of such ($i$, $j$) pairs.


To overcome the above limitations, we design a novel topology-aware protocol for asynchronous training. At every iteration, each agent only pairs with one of its neighbors for parameter exchange and update. An extra checking is conducted to guarantee that the paired agents are qualified for the topology-aware noise reduction: two agents cannot be paired twice in two consecutive iterations. Otherwise, the aggregated parameter selected by one agent in the previous iteration is not a secret to the other agent, and full-scale noise has to be added in this iteration. Hence the topology-aware noise reduction cannot be applied. Specifically, during the training process, agent $i$ randomly selects a neighbor $j$ which is different from the paired neighbor in the previous iteration. Then $i$ asks $j$'s availability for parameter sharing. If $j$ agrees to collaborate with $i$ in this iteration, they exchange parameters with the reduced noise scale and update the models following Eq. \ref{eq:aggre-red}. If agent $i$ cannot find a qualified or available pair at this iteration, it will update its estimate by itself.


Algorithm \ref{alg:asynchronous} describes the detailed steps of our asynchronous learning protocol. Similar to Algorithm \ref{alg:synchronous}, it takes the same parameters as input, and updates the initial estimate (Lines \ref{line:abroadcast}-\ref{line:ainitial}). At the $t$-th iteration, agent $i$ passively waits for pairing request from other neighbors. Meanwhile, it also actively searches in a random order for a neighbor that is not paired with it in the previous iteration (Lines \ref{line:search-begin}-\ref{line:search-end}). If the selected agent $j$ is available or $i$ receives a pairing request from $j'$, $i$ stops searching and pairs with $j_i^{t+1} = j$ ($j'$). Agent $i$ sends $\widetilde{x}_{i}$ to and receives $\widetilde{x}_{j_i^{t+1}}$ from $j_i^{t+1}$ (Line \ref{line:asend}). Then, $i$ adopts time-aware and topology-aware strategies to reduce the noise scale (Lines \ref{line:adecay}-\ref{line:atopology}) and updates its estimates (Lines \ref{line:agradient}-\ref{line:aupdate}). Otherwise, $i$ only utilizes the time-aware noise decay to update its estimate locally (Lines \ref{line:sdecay}-\ref{line:supdate}).



\begin{algorithm}[t!]
    \SetAlgoLined
    \SetKwInOut{Input}{Input}
    \SetKwInOut{Output}{Output}
    \Input{Initial estimate $x_0$, initial budget $\sigma_{0,i}$, learning rate $\lambda$, number of iterations $T$}
    Send $\sigma_{0, i}$ to $\mathcal{N}_i$ and receive $\{\sigma_{0,j}\}_{j \in \mathcal{N}_i}$ \label{line:abroadcast}\;
    $\bar{g}(x_0, \xi_s)$ ~$\gets$~ Compute the clipped gradient \;
    $\widetilde{x}_{i} \gets x_0 - \lambda\bar{g}(x_0, \xi_s) + G(\sigma_{0,i}^2C^2)$ \label{line:ainitial}\;
    $j_i^0 \gets None$\;
    \For{$t \in [0, T)$}{
        $\sigma_{t, i} \gets \texttt{Decay}(\sigma_{0, i}, t)$\;
        $\mathcal{N^*} \gets \mathcal{N}_i/j_i^t$ \label{line:search-begin} \;
        \While{$\mathcal{N^*}\neq \emptyset$}{
            Randomly select $j \in \mathcal{N^*}$ and $\mathcal{N^*} \gets \mathcal{N^*} \backslash j$\;
            Ask if $j$ is available for pairing up\;
            \If{$j$ is available}{
                $j_i^{t+1} \gets j$\;
                $\mathcal{N^*} \gets \emptyset$\;
            }
            \If{receive pairing request from $j'$}{
                $j_i^{t+1} \gets j'$\;
                $\mathcal{N^*} \gets \emptyset$ \label{line:search-end} \;
            }
        }
        \If{$j_i^{t+1}$ is found}{
            Send $\widetilde{x}_{i}$ and receive $\widetilde{x}_{j_i^{t+1}}$ to/from $j_i^{t+1}$ \label{line:asend}\;
            $\sigma_{t, j_i^{t+1}}\gets \texttt{Decay}(\sigma_{0, j_i^{t+1}}, t)$ \label{line:adecay}\;
            $\sigma_{t, i} \gets \sqrt{\sigma_{t, i}^2 - (1- \alpha)^2\sigma_{t, j_i^{t+1}}^2}$\label{line:atopology}\;
            $\bar{g}(\widetilde{x}_{i}, \xi_s)$ ~$\gets$~ Compute the clipped gradient\label{line:agradient}\;
            Update the local estimate $\widetilde{x}_{i} \gets \alpha\widetilde{x}_{i} + (1-\alpha)\widetilde{x}_{j_i^{t+1}}^i  - \lambda\bar{g}(\widetilde{x}_{i}, \xi_s) + G(\sigma_{t, i}^2C^2)$ \label{line:aupdate}\;
        }\Else{
            $\sigma_{t, i}\gets \texttt{Decay}(\sigma_{0, i}, t)$ \label{line:sdecay}\;
            $\bar{g}(\widetilde{x}_{i}, \xi_s)$ ~$\gets$~ Compute the clipped gradient\;
            Update the local estimate $\widetilde{x}_{i} \gets \alpha\widetilde{x}_{i} + (1-\alpha)\widetilde{x}_{j_i^{t+1}}^i  - \lambda\bar{g}(\widetilde{x}_{i}, \xi_s) + G(\sigma_{t, i}^2C^2)$ \label{line:supdate}\;
        }

    }
    \Return $\widetilde{x}_{i}$
    \caption{Differentially private decentralized learning for agent $i$ in the asynchronous mode.}\label{alg:asynchronous}
\end{algorithm}

\section{Theoretical Analysis}
\label{sec:analysis}
We perform a formal analysis about Algorithms \ref{alg:synchronous} and \ref{alg:asynchronous} from the aspects of privacy and efficiency.

\subsection{Proof of DP}
First, we prove Algorithm \ref{alg:synchronous} is differentially private by carefully choosing the initial noise parameters. We track the accumulated privacy loss of the training process using R\'enyi DP \cite{mironov2017renyi}, which is a natural relaxation of DP based on the R\'enyi divergence and allows tighter analysis of tracking cumulative privacy loss and ensures a sublinear loss of privacy as a function of the number of iterations.

\begin{theorem}\label{theorem}
    Let the number of iterations be $T$. For any decentralized system $\mathcal{G}$ and every agent $i \in V$, the randomized mechanisms in Algorithm \ref{alg:synchronous} is ($\epsilon_i, \delta_i$)-DP if we choose
    \begin{equation}\label{eq:theorem}
        \sigma_{0,i} \geq \frac{8\sqrt{T\log\frac{1}{\delta_i}\log\frac{1.25}{\delta_i}} }{\epsilon_i |D_i|}
    \end{equation}
\end{theorem}
\begin{proof}
We prove the theroem in the synchronous mode and ignore the time-aware noise decay strategy since it does not incur any additional privacy loss \cite{yu2019differentially}. We clip the gradients in $l_2$ norm of $C$ and assume the privacy budget $\epsilon_{i}'$ is the same at each iteration. According to the Gaussian mechanism \cite{dwork2006our}, the update rule in Line \ref{line:update} is ($\epsilon_{i}', \delta_i$)-DP at one iteration if we choose $$\sigma_{0, i} \geq \frac{\sqrt{2\log\frac{1.25}{\delta_i}}}{\epsilon_{i}'|D_i|}.$$

Using R\'enyi composition theorem \cite{mironov2017renyi}, our new update rule is ($\epsilon_i, \delta_i$)-DP after $T$ iterations if we choose $$\epsilon_i = 4\epsilon_i'\sqrt{2T\log\frac{1}{\delta_i}}.$$ Then, we have $$\epsilon_i' = \frac{\epsilon_i}{4\sqrt{2T\log\frac{1}{\delta_i}}}.$$ Combining the above equations, we conclude that our update rule in Line \ref{line:update} is ($\epsilon_i, \delta_i$)-DP if we choose $\sigma_{0,i}$ such that

\begin{equation}
\label{eq:prove2}
\sigma_{0,i} \geq \frac{8\sqrt{T\log\frac{1}{\delta_i}\log\frac{1.25}{\delta_i}} }{\epsilon_i |D_i|}
\end{equation}

We have proven that the local estimate of agent $i$ is differentially private during the training process. Then, we prove that for $\forall j \in \mathcal{N}_i$, the estimates generated for $j$ is also differentially private. Let $k, (k,j) \notin E$ be the selected agent for generating estimate for $j$. Since $j$, $k$ are not directly connected, the noise of $\widetilde{x}_{k}^i$ can be used as a random component to guarantee the DP of $i$ against $j$. Thus, because all agents generate noise independently, the noise scale for $j$ should satisfy

\begin{equation}
\label{eq:prove1}
G(\sigma_{0,i}) = G(\sigma_i^j) + (1-\alpha)G_{0,k}
\end{equation}

According to the additivity of Gaussian distribution, the noise parameter for the estimate for $j$ is $$\sigma_{i}^j = \sqrt{\sigma_{0, i}^2 - (1- \alpha)^2\sigma_{0, k}^2}.$$ Therefore, in Algorithm \ref{alg:synchronous}, the estimates generated for the neighbors of agent $i$ are also differentially private.
\end{proof}

The DP of Algorithm \ref{alg:asynchronous} can also be analyzed in a similar way. Note that an agent cannot pair with another agent twice in a row. Therefore, even the agents in a decentralized system are fully connected, the topology-aware noise reduction still works in such situation, where Algorithm \ref{alg:synchronous} fails.

\subsection{Efficiency Analysis of \AlgName}
\rtwo{Our protocols can reduce the noise and thus improve the usability of the trained models using the proposed \AlgName algorithm when considering the communication topology. Here, we theoretically analyze the efficiency of \AlgName by comparing the amount of added noise with and without \AlgName. Without loss of generality, we assume $$\sigma =\sigma_{t,i} = \sigma_{t, j} \ \text{for}  \ \forall i, j \in V \ \text{and} \ (i, j) \in E.$$ Let $\sigma_{t, i}^j$ be the noise parameter of $\widetilde{x}_{i}^j$ at iteration $t$. According to the proposed topology-aware noise reduction strategy,
\begin{align}\nonumber
    \sigma_{t, i}^j &= \sqrt{\sigma_{t, i}^2 - (1- \alpha)^2\sigma_{t, k}^2} \\
                    &= \sigma\sqrt{2\alpha -  \alpha^2}.
\end{align}
}

\rtwo{Compared with the full-scale noise parameter, the noise added to $\widetilde{x}_{i}^j$ is reduced by a factor of $\sqrt{2\alpha -  \alpha^2}$. We can observe that $\sigma_{t, i}^j$ decreases as $\alpha \in (0, 1)$ decreases.  When $\alpha$ approaches 0, the noise of the estimates that agent $i$ sends to/receives from its neighbors would be significantly reduced. Thus, the usability of the trained models would be theoretically improved because of the decrease of the added noise.}

\rtwo{In synchronous mode (Algorithm \ref{alg:synchronous}), agent $i$ can always reduce the noise of the estimates for its neighbor $j$ using the \AlgName if there exists an agent that connects to agent $i$ and cannot communicate with $j$ directly, i.e., $$\mathcal{N}_i / \mathcal{N}_j \neq \emptyset.$$ For the asynchronous settings (Algorithm \ref{alg:asynchronous}), \AlgName works if it finds a pairing neighbor during the iteration. Therefore, the agents in both synchronous and asynchronous modes can theoretically improve the utility of their trained models using our \AlgName.}

%% file: body/experiments.tex
\section{Experiments}\label{sec:experiments}
\subsection{Implementation and Experimental Setup}
\bheading{Dataset and DNN model.}
We conduct experiments mainly on the MNIST dataset. It consists of a training set of 60k samples and a test set of 10k samples. We consider a fully connected network with a hidden layer of size 100 for image classification. We set a fading learning rate $\lambda$ with the initial value of 0.05. Our solution is general and can be applied to other DNN tasks as well (as demonstrated in Section \ref{sec:cifar}).

For the implementation of the decentralized system, we consider a network consisting of 30 agents, and each agent connects to others with the probability of 0.2 (connection rate). This decentralized system is guaranteed to be fully connected, i.e., there exists at least one path connecting two arbitrary agents. The training set of each agent is independent and identically distributed with the same size. In the synchronous mode, all 30 agents participate at each training iteration. In the asynchronous mode, we assume 10\% of random agents will not be involved at each iteration.

Without loss of generality, the agents have same privacy budget (1.0) and relaxation hyper-parameter ($10^{-5}$). We assume the agents reach the consensus on
the time-aware noise decay strategy, where $\gamma$ and $period$ are 0.9 and 1000, respectively. We clip the gradients in $l_2$ norm of 4.0.

\bheading{Baselines and metrics.}
We consider different decentralized learning algorithms in our experiments:

\begin{itemize}
    \item \emph{No Noise}: the agents exchange parameters without DP protection.
	\item \emph{Li18}: the DP-SGD algorithm proposed by Li et al. \cite{li2018differentially}.
	\item \emph{Li18+MA}: we integrate Li18 with moments account \cite{abadi2016deep} to track the accumulated privacy loss.
	\item \ronetwo{\emph{UDP}: the user-level DP-SGD algorithm proposed by Wei et al. \cite{wei2021user}}.
	\item \ronetwo{\emph{Optimal}: the optimal DP analysis for SGD proposed by Asoodeh et al. \cite{asoodeh2021three}.}
	\item \emph{Proposed}: our proposed learning protocols.
\end{itemize}

It is worth noting the first five solutions cannot be applied to the asynchronous mode directly. For fair comparisons, we modify their update rules as Eq. \ref{eq:aggre-dp} to follow our learning protocol for asynchronous learning. For each algorithm, we measure the testing accuracy of each agent's model at every iteration during the training, and report the average accuracy.

\subsection{Effectiveness of \AlgName}
\label{sec:effective}
We evaluate and compare the performance of those DP-SGD algorithms under different settings in both synchronous and asynchronous modes.


\bheading{Epoch v.s. accuracy.}
Fig. \ref{fig:sync_acc} illustrates the trend of average testing accuracy in the training process with different $\alpha$ values. First, we observe that our proposed algorithm outperforms all baselines, and is closer to the No Noise case, for different $\alpha$ values and modes. Such advantage is more obvious with a smaller $\alpha$, as the reduced noise is larger. Second, Li18+MA has higher performance than Li18 because of the usage of MA. \ronetwo{With the new DP technique, Optimal outperforms UDP and Li18+MA in most settings.} Different from our solution, the usability of the models from all baselines significantly decreases as $\alpha$ decreases. This is caused by the increase of the noise of the selected estimates.
Third, the model training in synchronous mode converges slightly faster than the one in the asynchronous mode, since each agent can contribute to the model training to accelerate the process.

\bheading{Privacy budget v.s. accuracy.}
We consider the impact of privacy budget on the model accuracy, as shown in Fig. \ref{fig:budget}. We can observe our solution can beat the other DP solutions for different privacy budgets. Besides, when the privacy budget decreases, the model usability decreases, as more noise is required to inject to the estimates. Meanwhile, the advantage of our solution also increases, as the amount of reduced noise increases as well. This indicates that our algorithm is more effective when a small privacy budget is needed.



\begin{table*}[t!]\centering
	\resizebox{\textwidth}{!}{
		\begin{tabular}{c@{}c@{}c@{}c}
			\includegraphics[width=0.33\textwidth]{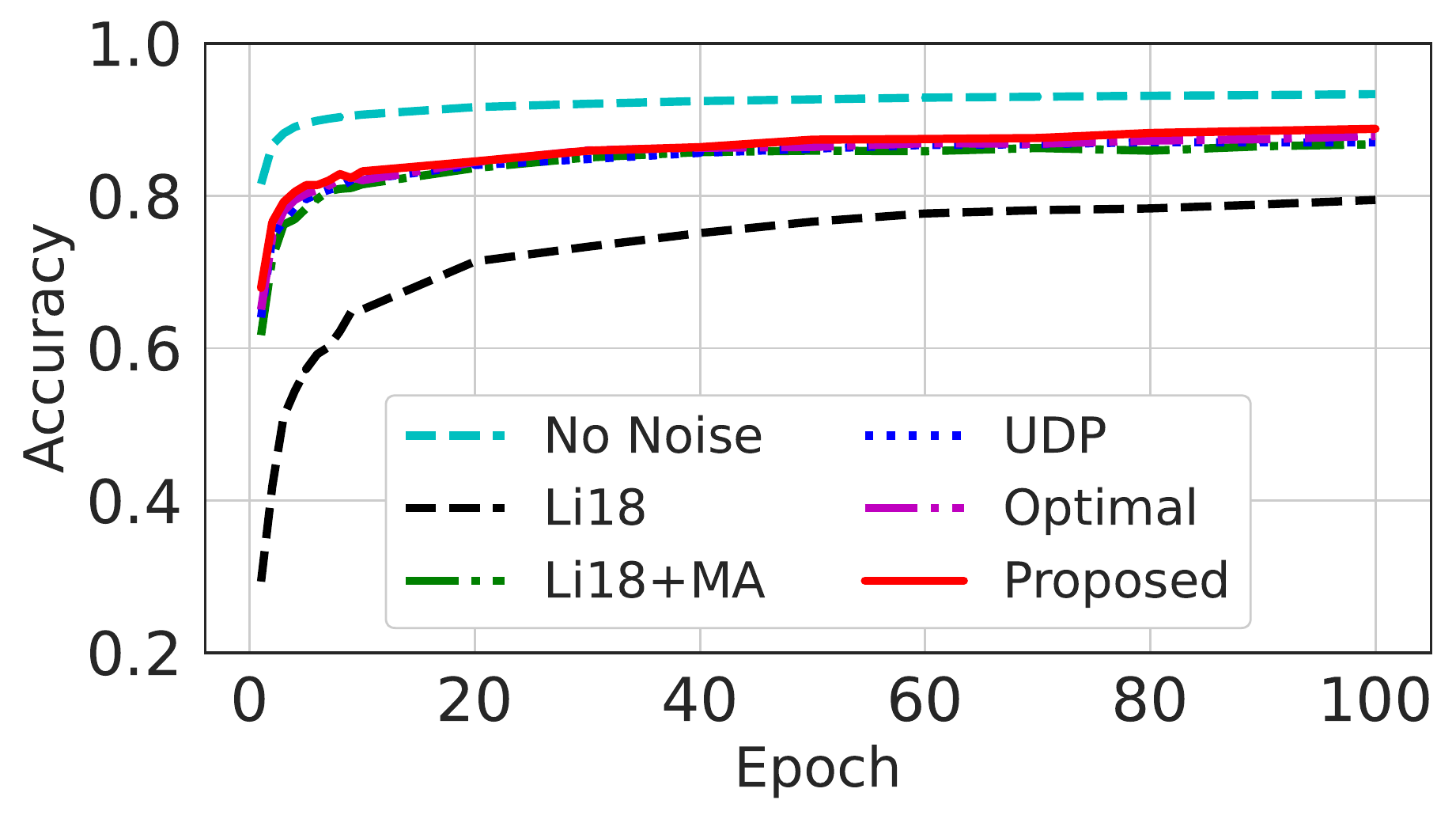}&\includegraphics[width=0.33\textwidth]{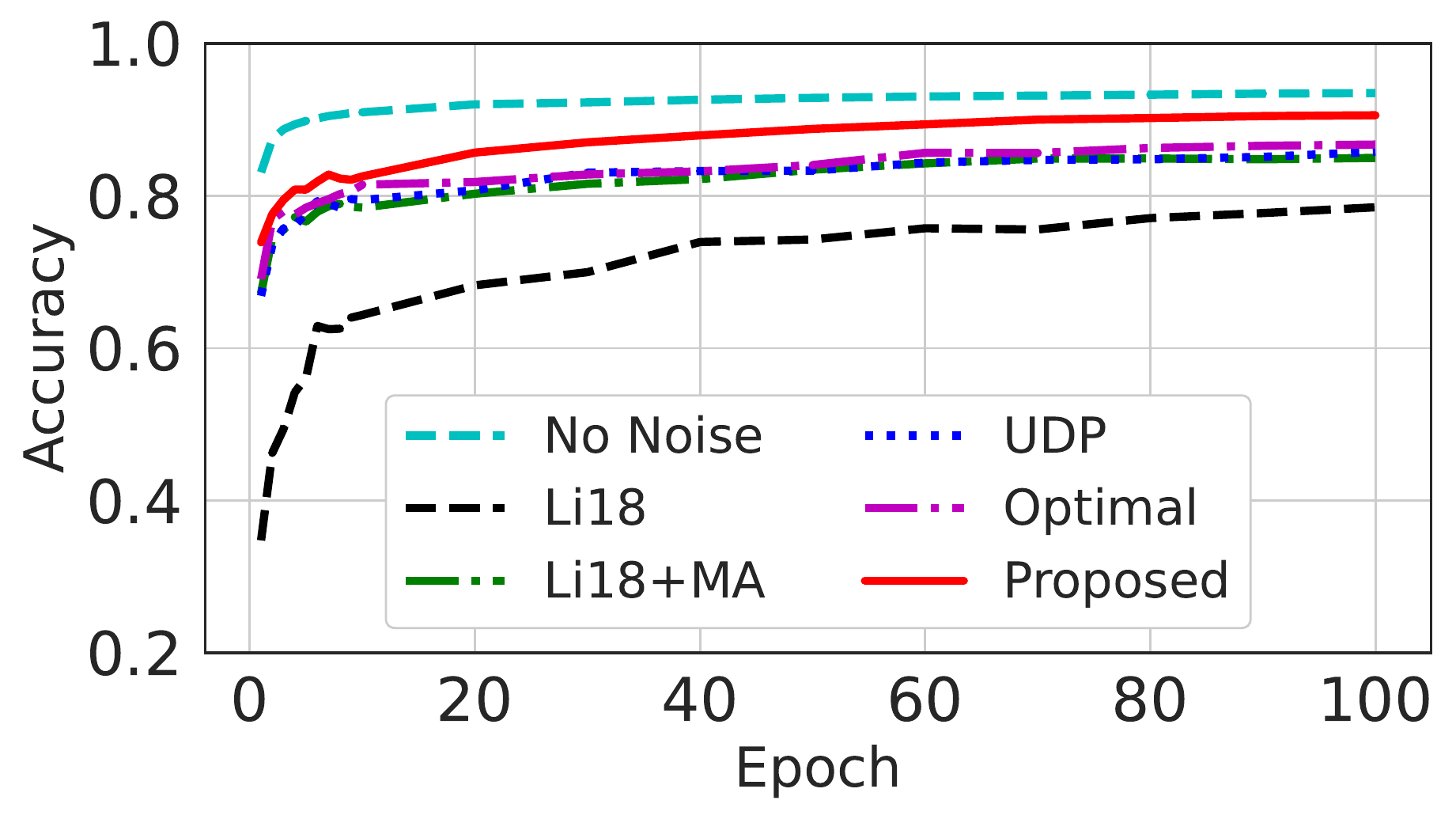}&\includegraphics[width=0.33\textwidth]{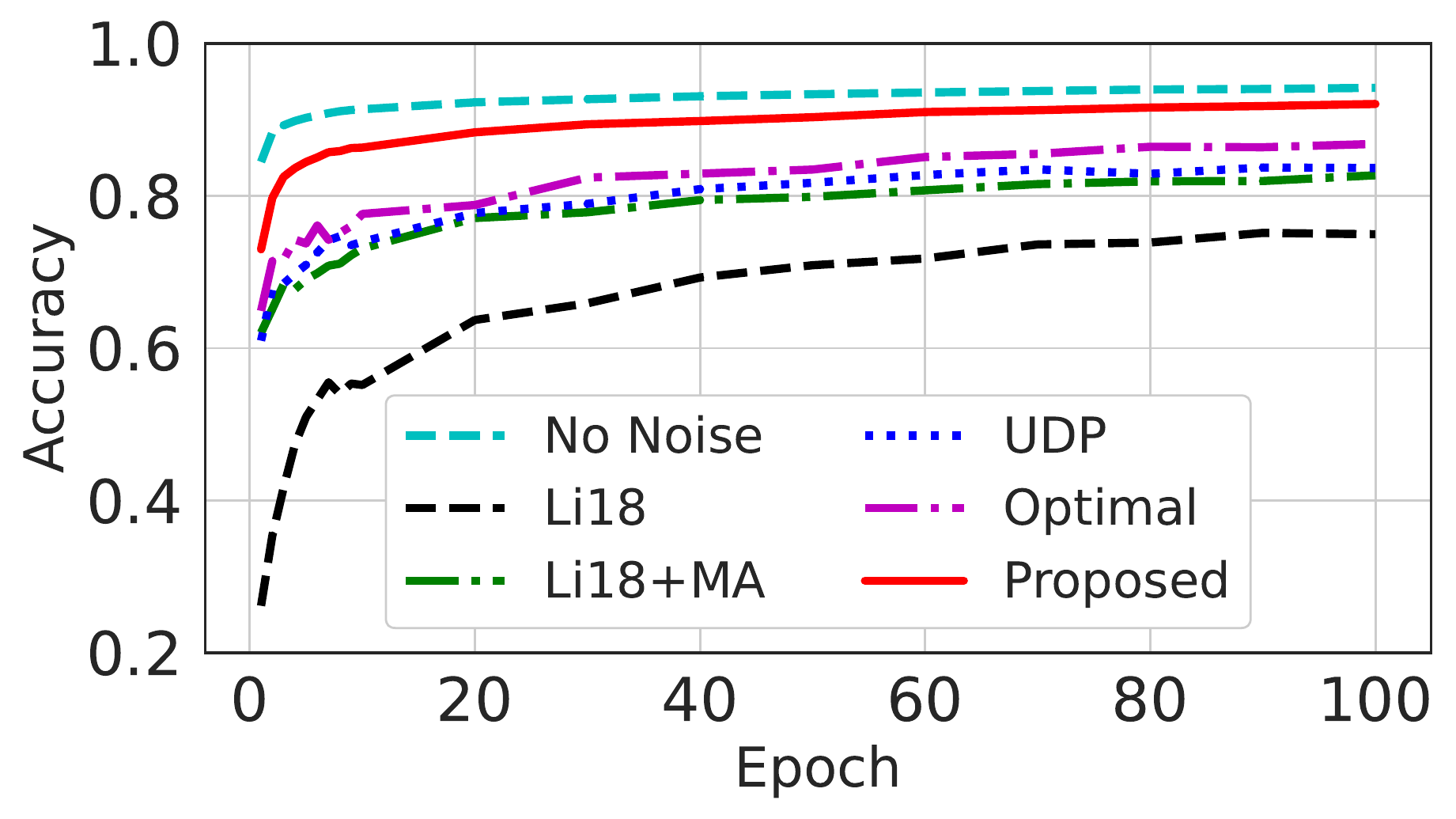}&\includegraphics[width=0.33\textwidth]{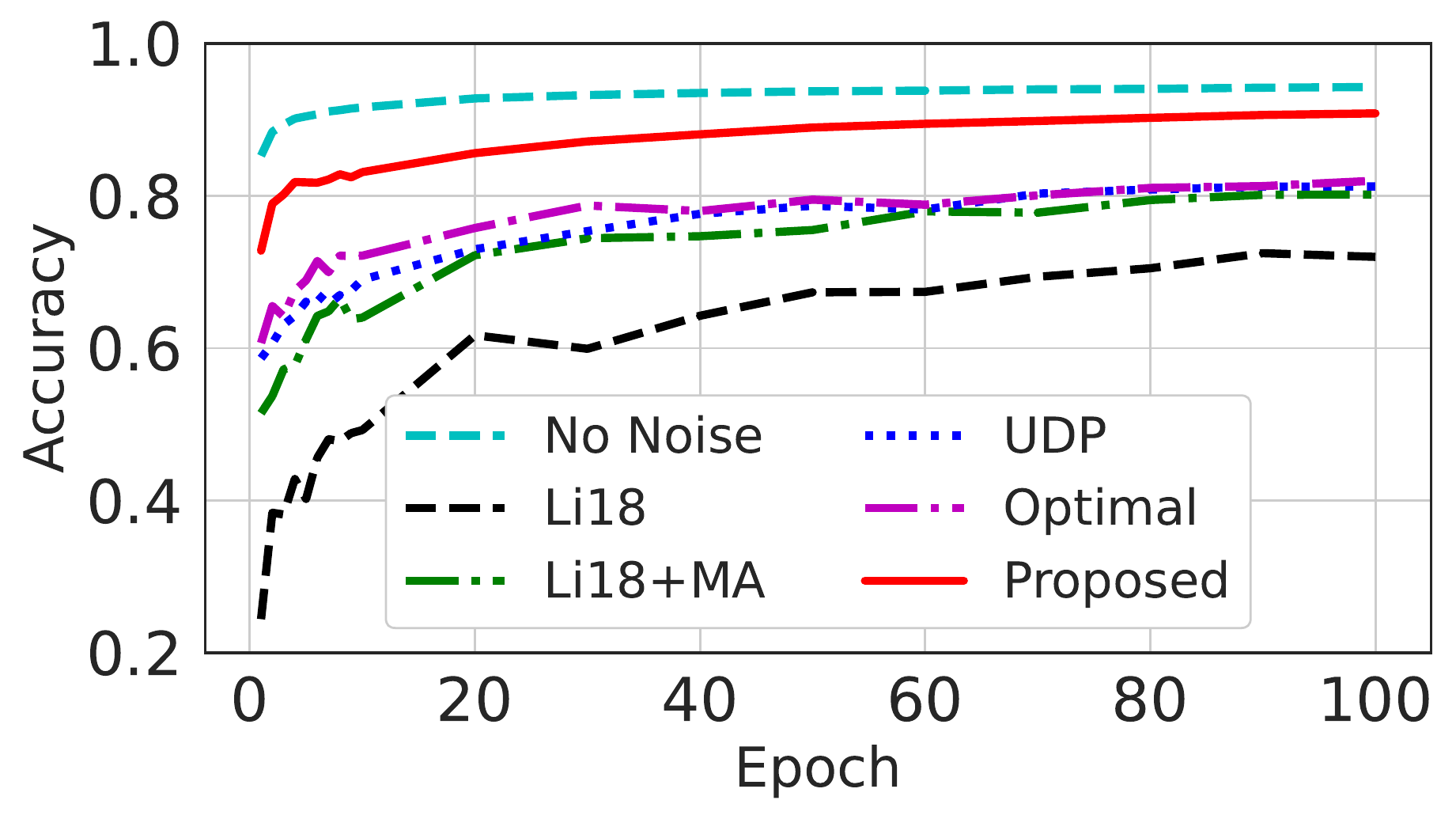}\\
			\includegraphics[width=0.33\textwidth]{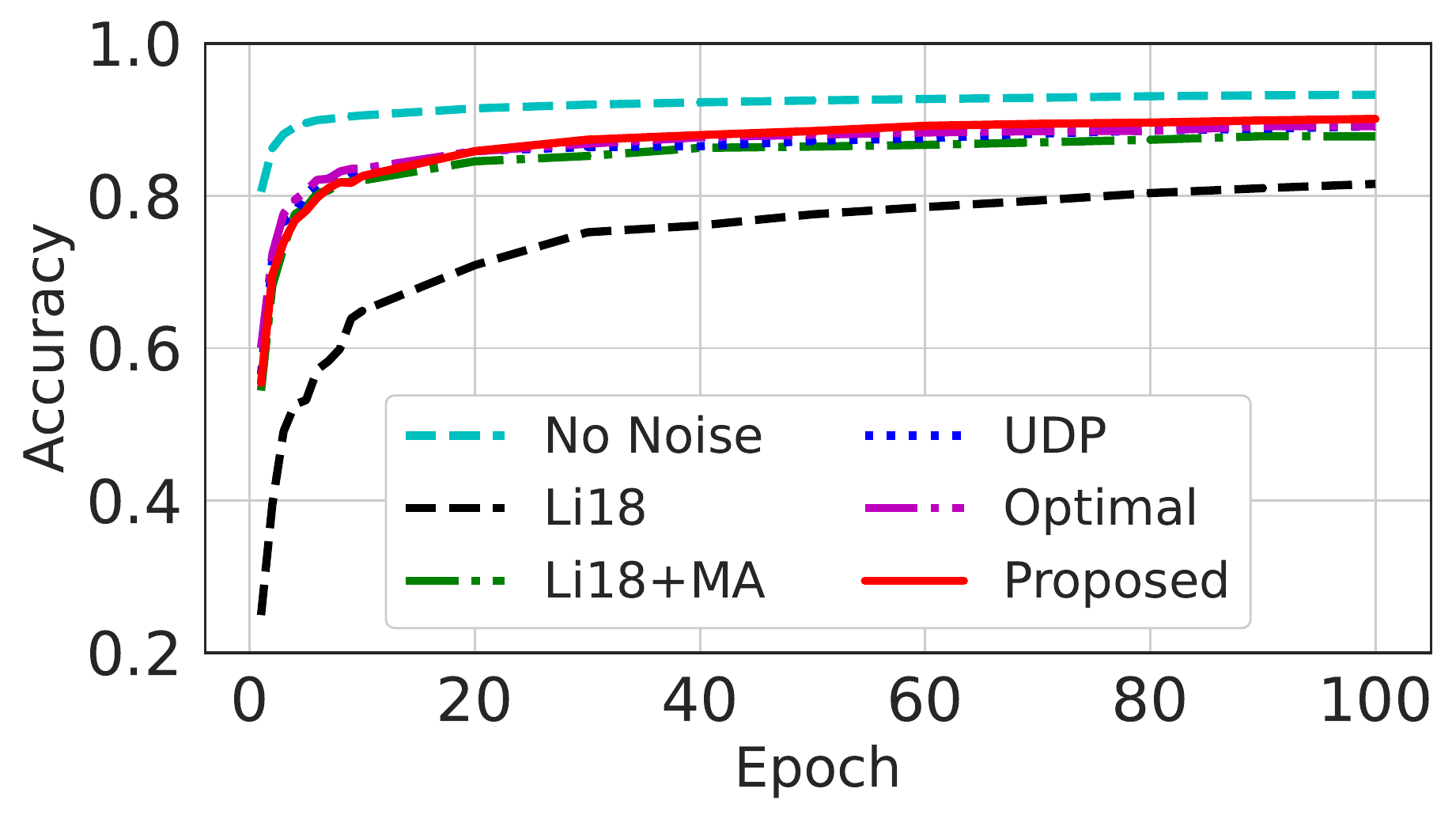}&\includegraphics[width=0.33\textwidth]{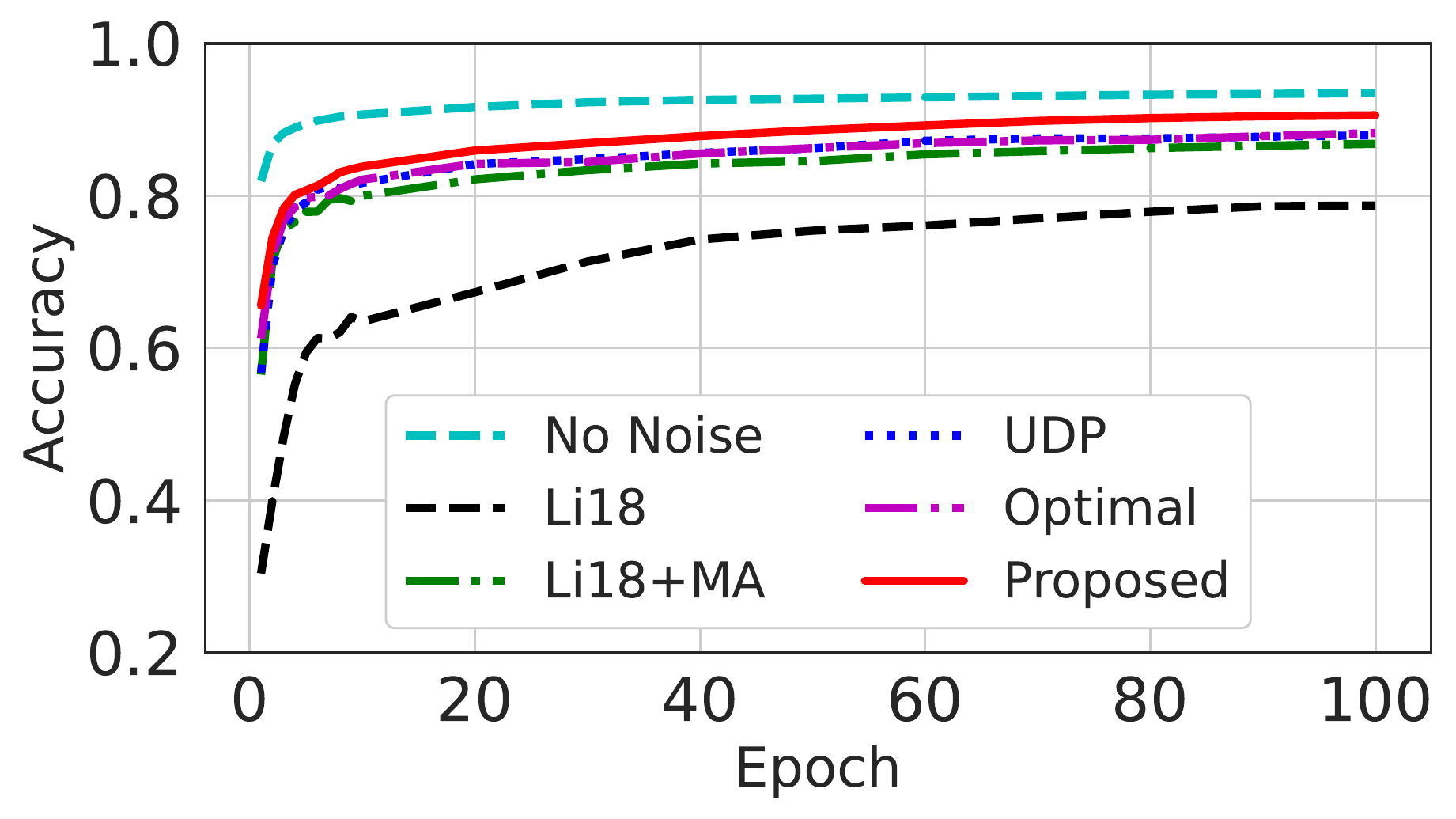}&\includegraphics[width=0.33\textwidth]{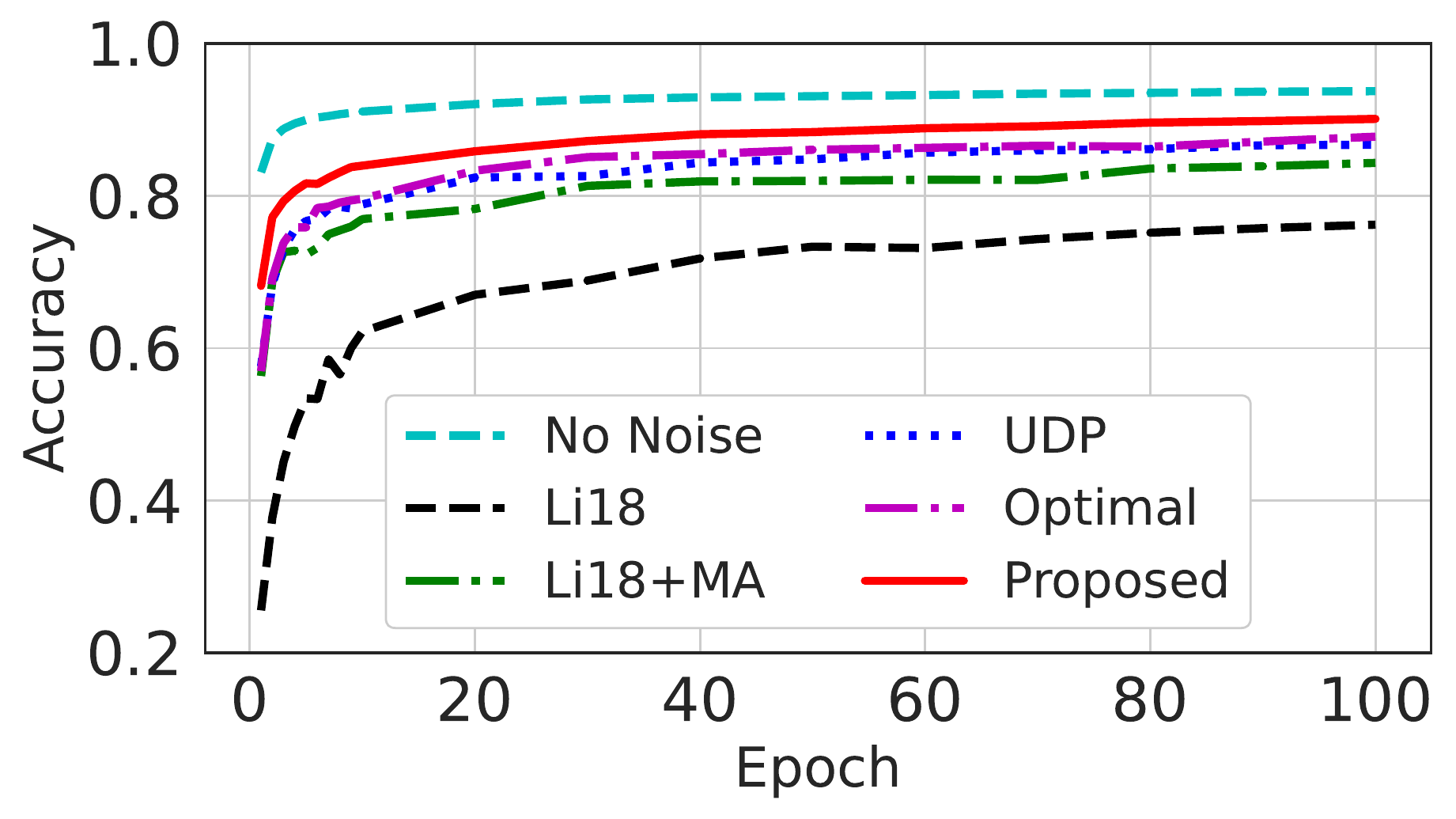}&\includegraphics[width=0.33\textwidth]{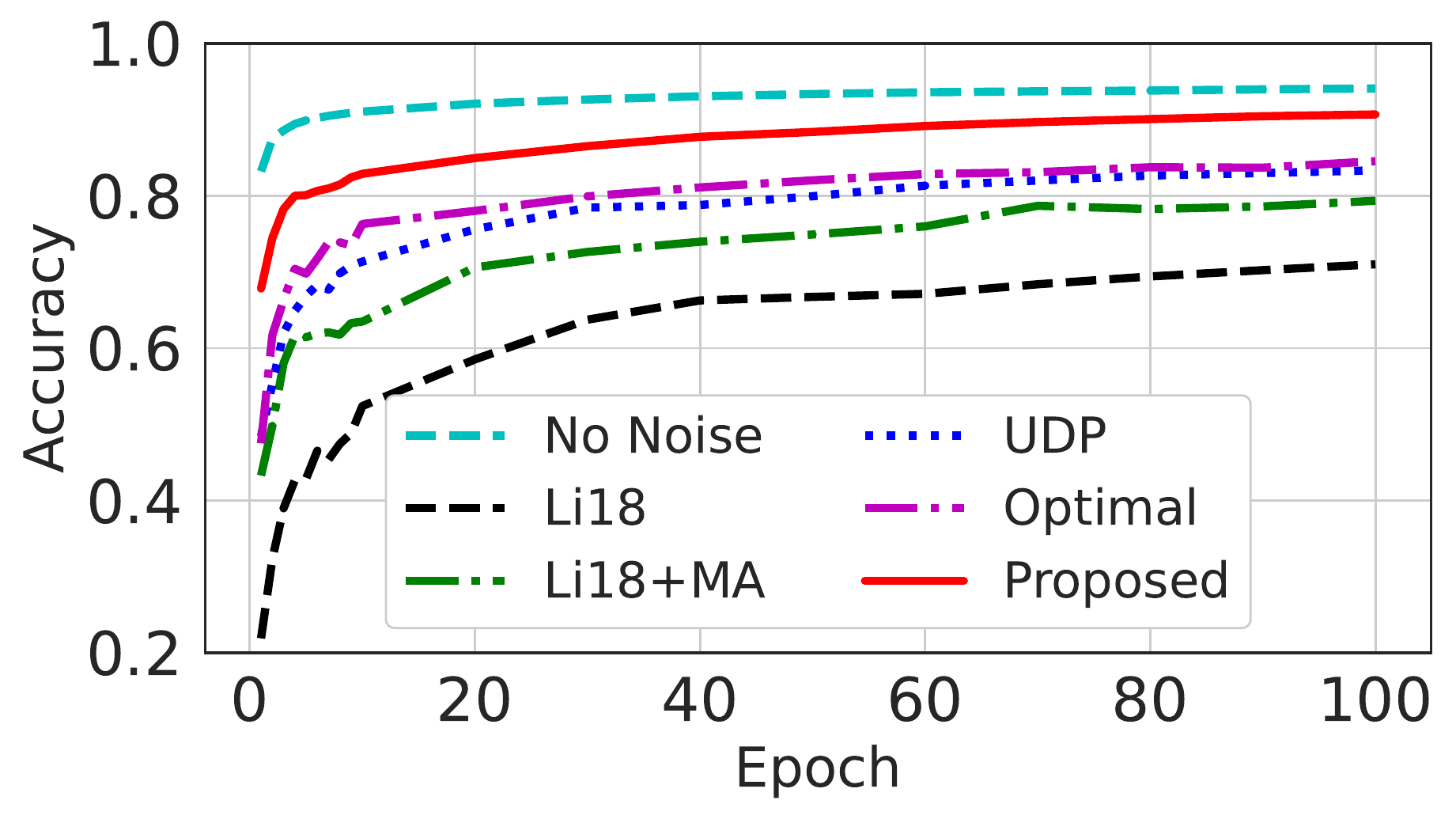}\\
			(a) $\alpha = 0.75$ & (b) $\alpha = 0.5$& (c) $\alpha = 0.25$ &(d) $\alpha = 0.125$
	\end{tabular}}
	\captionof{figure}{\ronetwo{The average accuracy of the agents with different $\alpha$ values under synchronous (first row) and asynchronous (second row) settings.}}\label{fig:sync_acc}
\end{table*}

\subsection{Impact of System Configurations}
\noindent\textbf{Connection rate.} We set the connection rate of the decentralized network as 0.2 in the previous experiments. Our proposed algorithm is effective under other connection rates as well. To validate this, we measure and compare the performance of different DP-SGD algorithms with the connection rates of 0.1 and 0.4. Without loss of generality, we consider the synchronous mode and set $\alpha$ as 0.25. Figure \ref{fig:rate} shows the average accuracy of the agents as the training epoch increases. We observe that the performance of each algorithm does not change with different connection rates. The underlying reason may be that although the number of an agent's neighbors is changed with the connection rate, the agent still selects one estimate for updates at each iteration. Then the training result will not be changed either. As such, our proposed solution can exhibit advantages over prior works under various network connection rates.

\noindent\ronetwo{\textbf{Number of agents.} We now investigate the impact of the number of agents on the performance of decentralized learning systems. We conduct experiments on decentralized systems with 40 and 50 agents in the synchronous mode. The experimental results are shown in Figure \ref{fig:number}. We observe that the accuracy of the trained models only slightly increases with more agents involved, indicating that the number of agents has a small positive impact on the decentralized systems.}

\begin{figure}[t]
	\centering
	\subfigure[Synchronous]{\includegraphics[width=0.49\columnwidth]{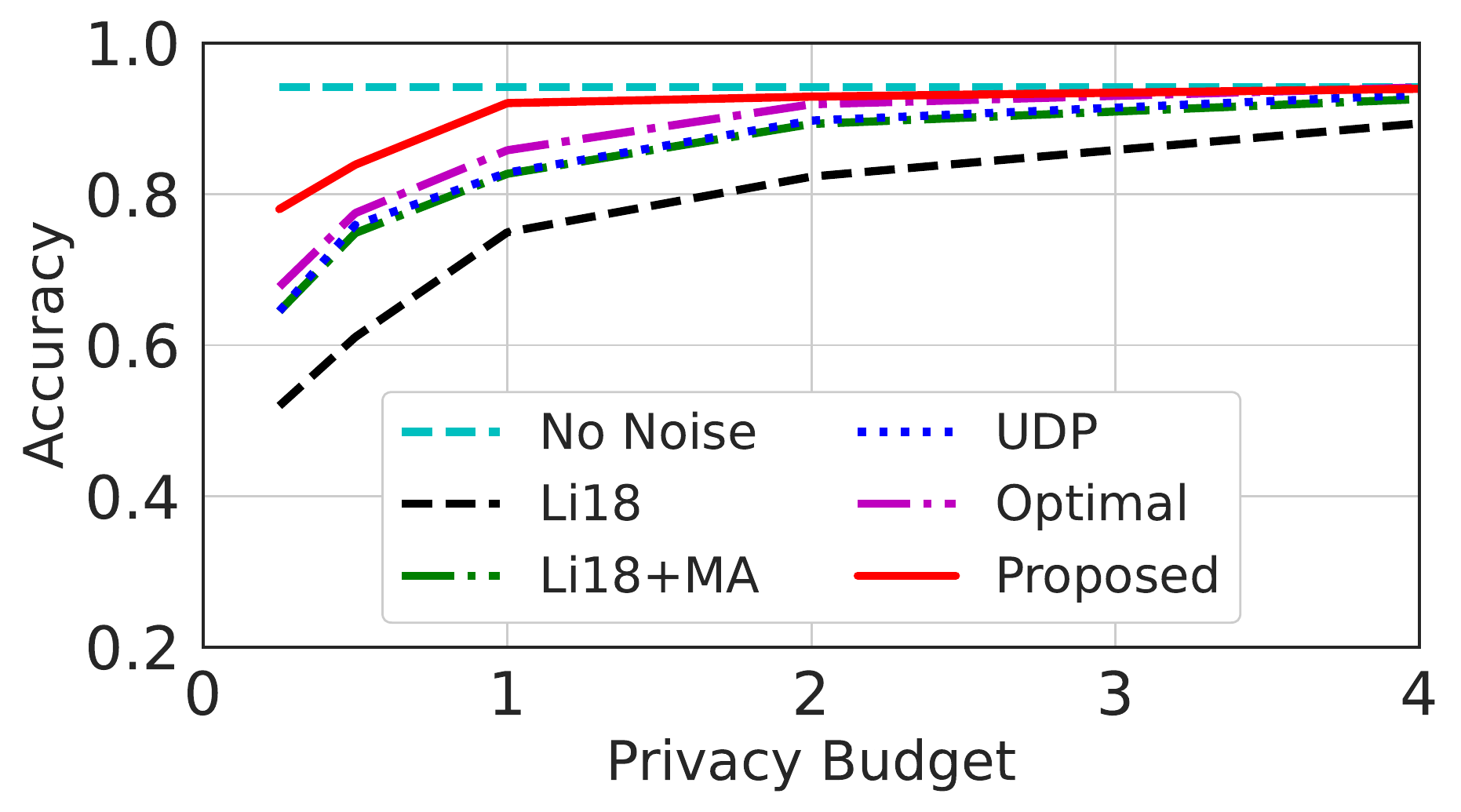}}
	\subfigure[Asynchronous]{\includegraphics[width=0.49\columnwidth]{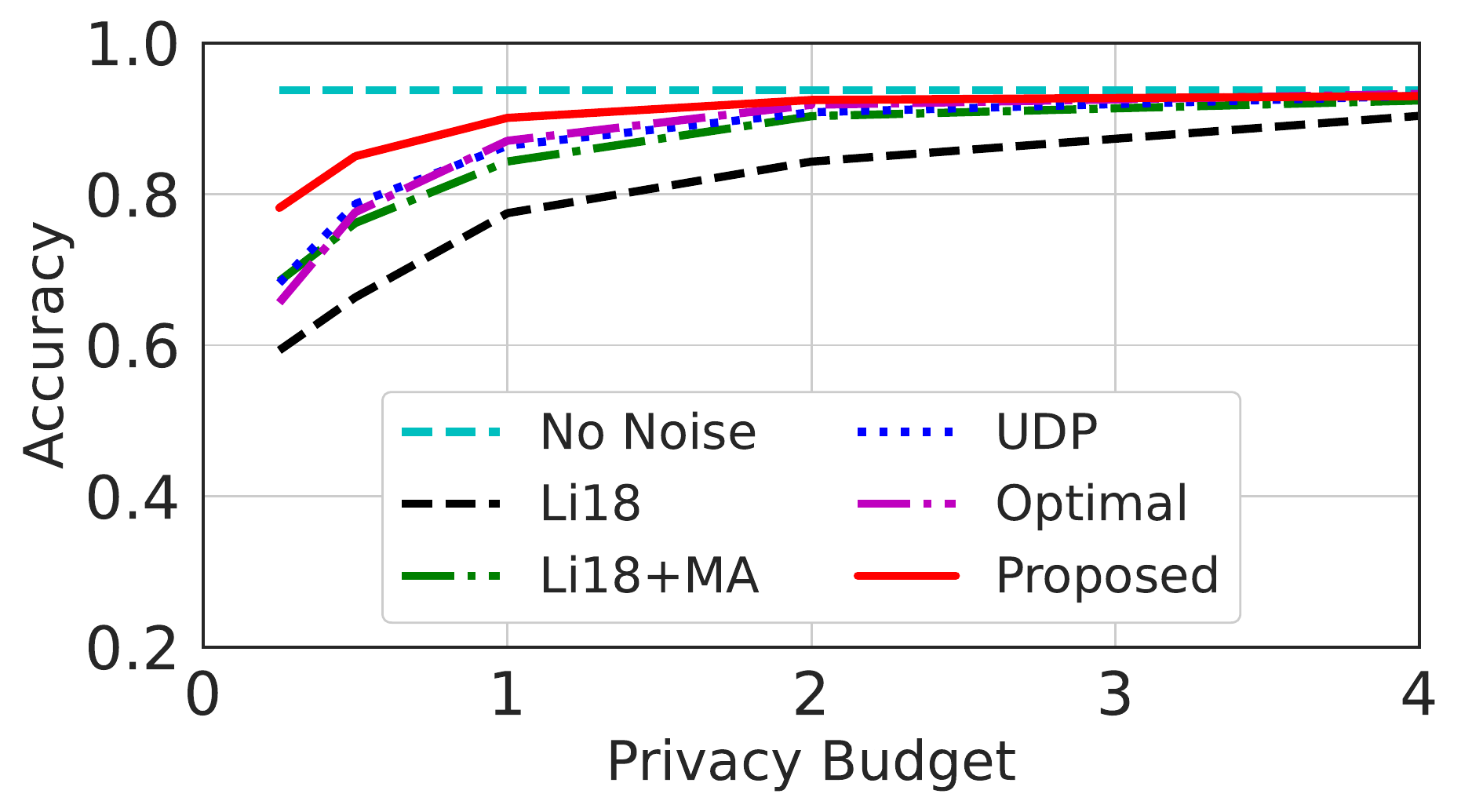}}

	\caption{\ronetwo{The average accuracy of the agents as the privacy budget increases.}}\label{fig:budget}
\end{figure}

\begin{figure}[t]
	\centering
	\subfigure[Connection Rate = 0.1]{\includegraphics[width=0.49\columnwidth]{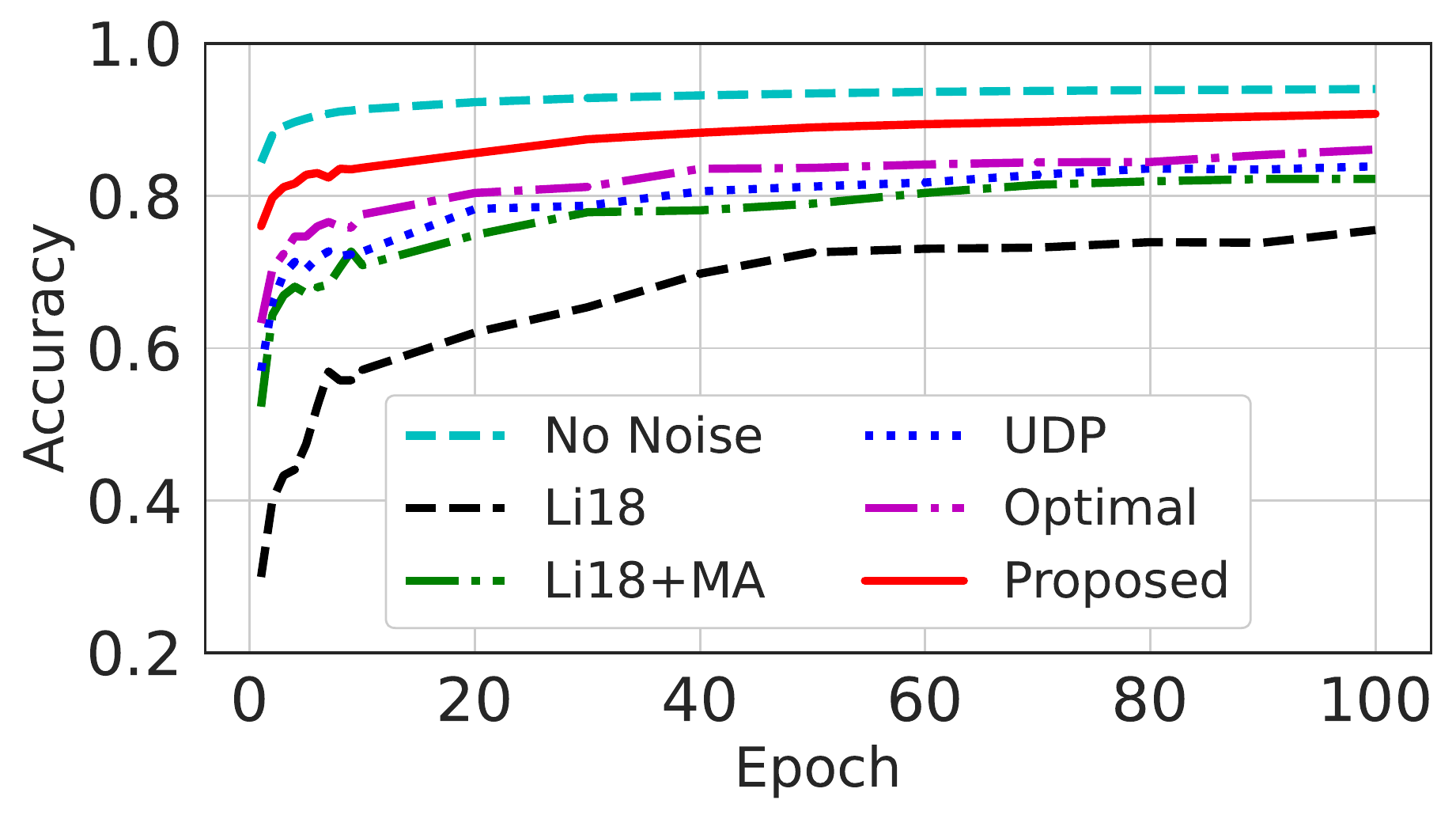}}
	\subfigure[Connection Rate = 0.4]{\includegraphics[width=0.49\columnwidth]{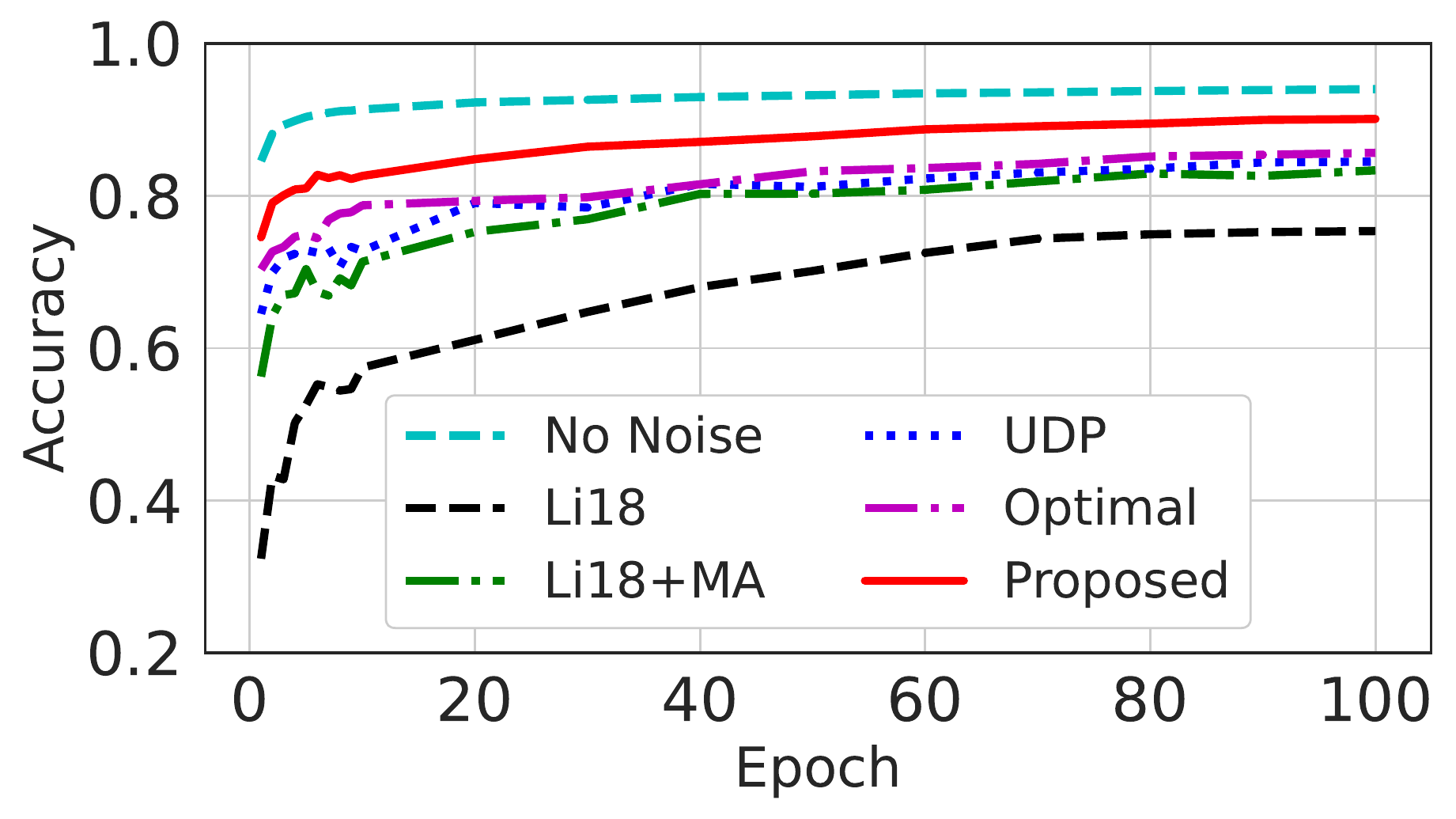}}
	\caption{\ronetwo{The average accuracy of the agents with different connection rates under the synchronous setting.}}
	\label{fig:rate}
\end{figure}

\noindent\ronetwo{\textbf{Parameters of the noise decay strategy.} We evaluate the impact of the parameters of the noise decay strategy on our \AlgName, i.e., $\gamma$ and $period$. In our experiments, $\gamma$ is set from 0.7 to 1.0 while $period$ varies from 8 to 12. Figure \ref{fig:impact} illustrates the average performance of the trained models under different parameter settings. Two observations are drawn. First, both $\gamma$ and $period$ have only limited impact on the performance of the decentralized learning systems, especially in Figure \ref{fig:impact} (a).  Second, as $period$ increases, the average accuracy of the trained models slightly decreases.}

\begin{figure}[t]
	\centering
	\subfigure[Number of Agents = 40]{\includegraphics[width=0.49\columnwidth]{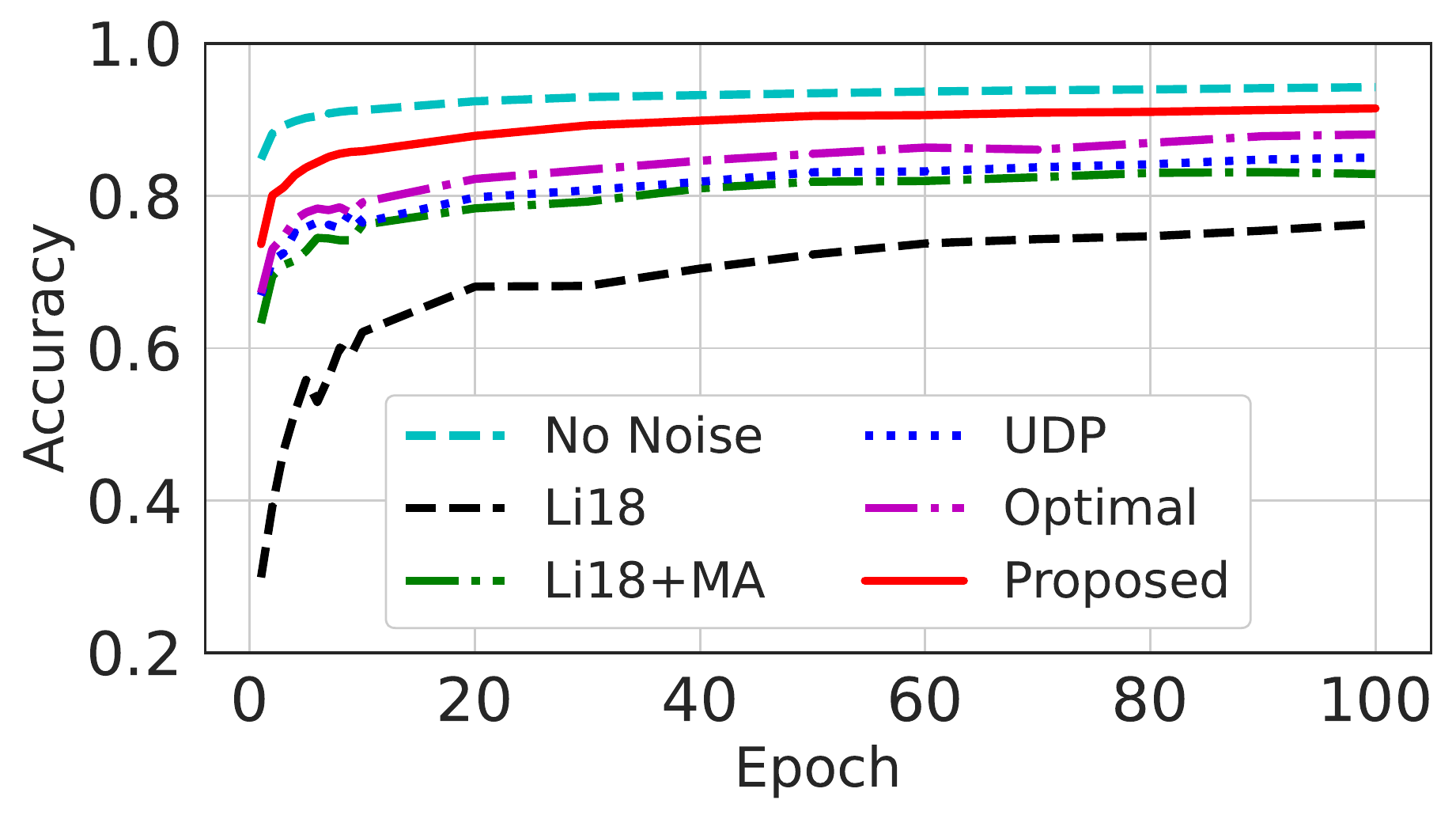}}
	\subfigure[Number of Agents = 50]{\includegraphics[width=0.49\columnwidth]{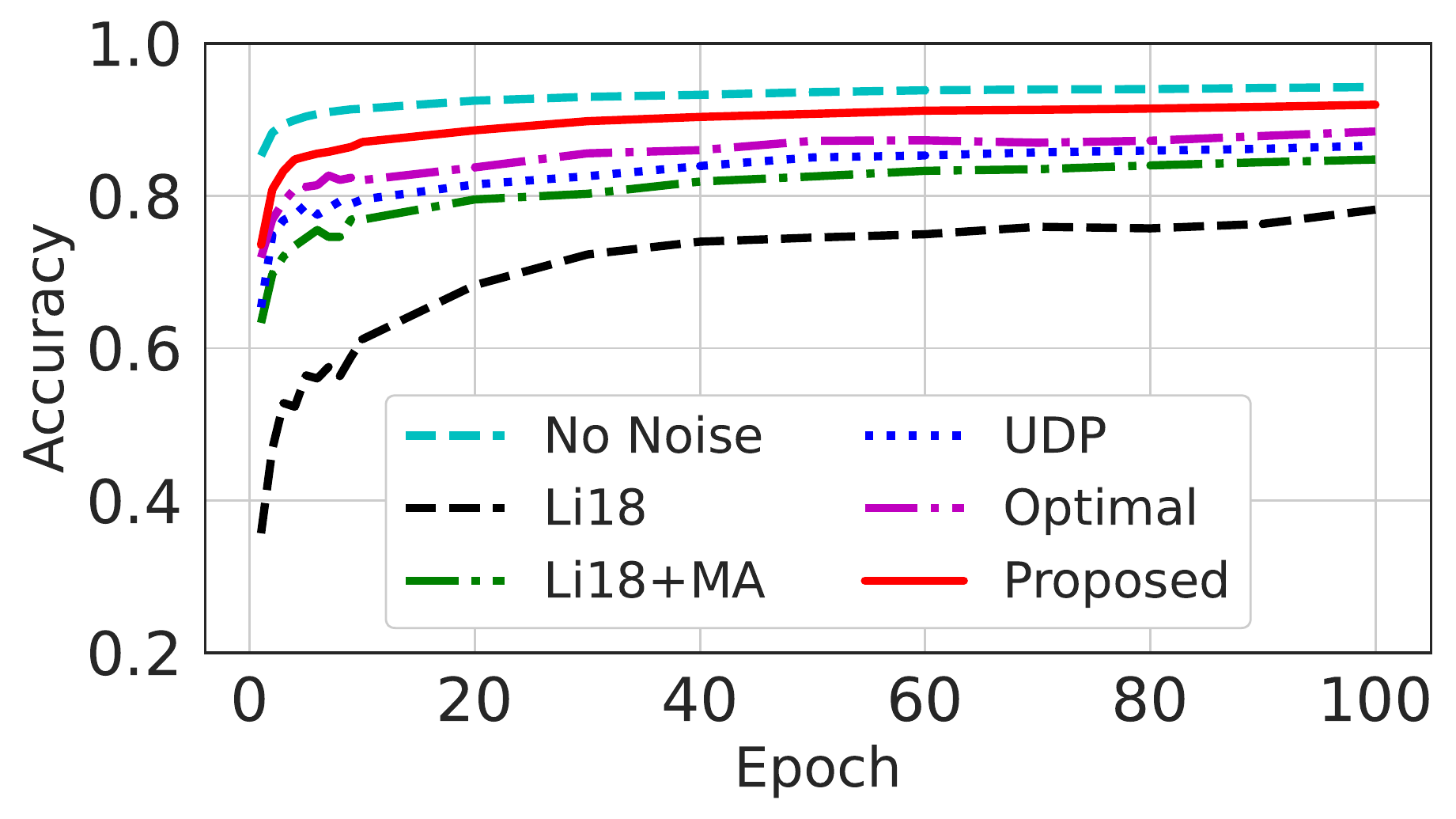}}
	\caption{\ronetwo{The average accuracy of the agents with different numbers of agents under the synchronous setting.}}
	\label{fig:number}
\end{figure}

\begin{figure}[t]
	\centering
	\subfigure[Impact of $\gamma$]{\includegraphics[width=0.49\columnwidth]{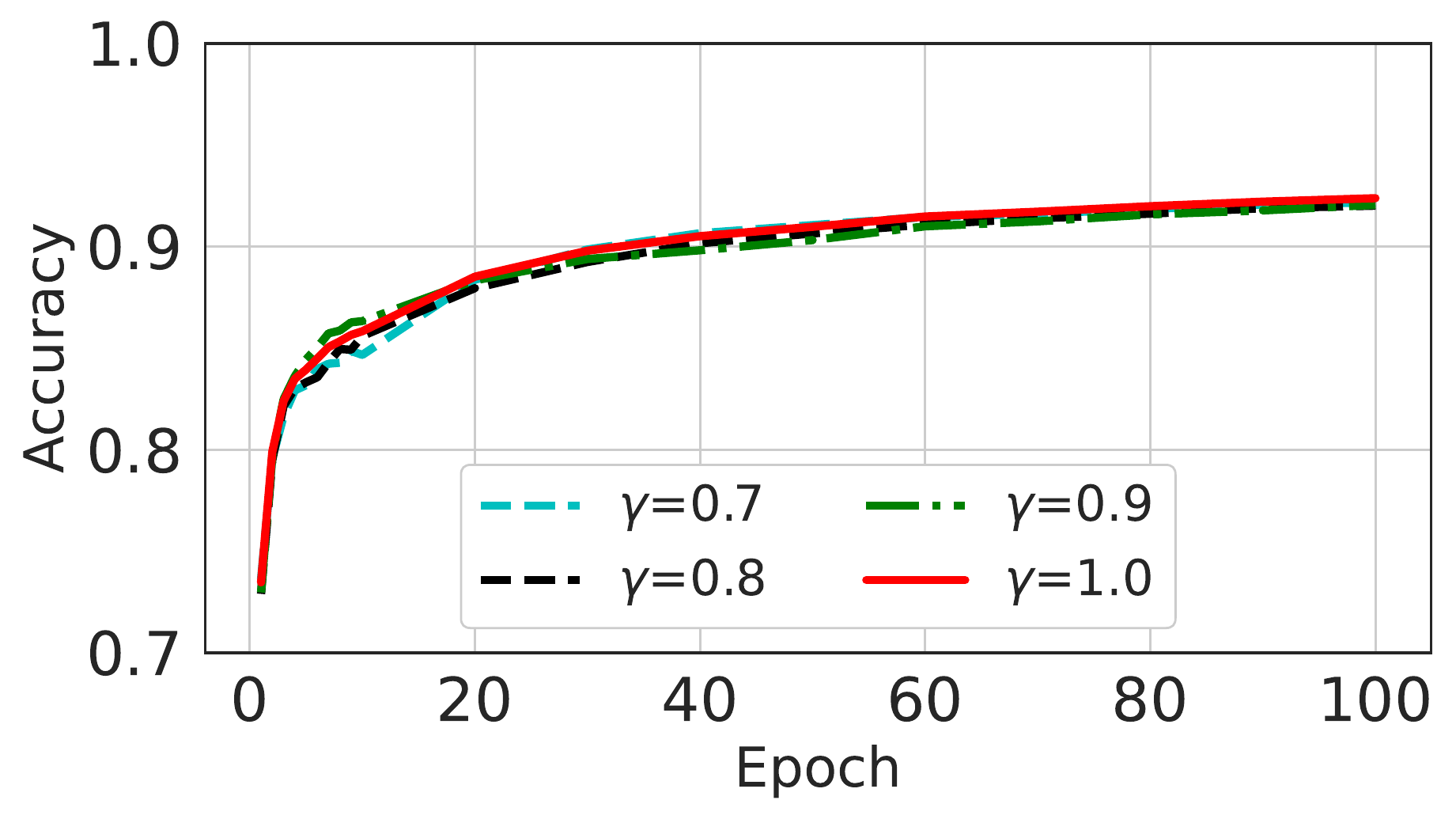}}
	\subfigure[Impact of $period$]{\includegraphics[width=0.49\columnwidth]{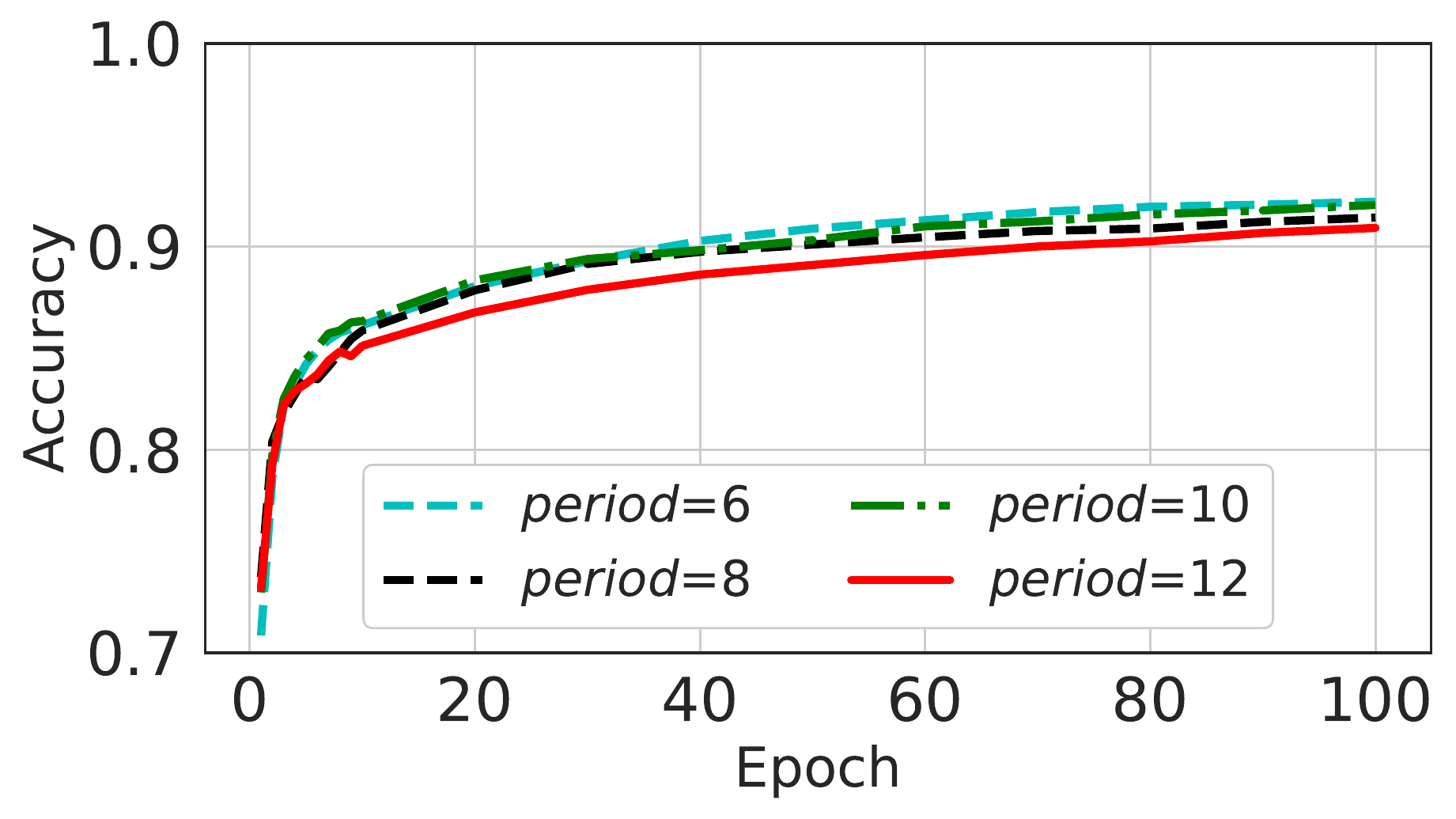}}
	\caption{\ronetwo{The average accuracy of the agents with different parameters of the noise decay strategy under the synchronous setting.}}
	\label{fig:impact}
\end{figure}

\noindent\textbf{Network topology.} We also evaluate our DP-SGD learning protocols on other typical network structures, \rone{such as the ring, star, tree, and mesh topologies. Figure \ref{fig:topology} illustrates the comparisons of decentralized network structures with different connections.
In each network topology, we set the number of total agents as 30.
Figure \ref{fig:topology_acc} shows the learning curves of different DP-SGD algorithms for both synchronous and asynchronous modes.}

We observe that in the synchronous mode, the average accuracy scores of our learning protocals are significantly higher than other baselines, which is attributed to our proposed topology-aware strategy. \ronetwo{The performance gap among Li18+MA, UDP, and Optimal is small due to the limitation of the corresponding DP optimizations.} In the asynchronous mode, our protocol is slightly better than others, although the advantage is not as big as the synchronous mode.

\begin{figure*}[t!]
	\centering
	\includegraphics[width=1.5\columnwidth]{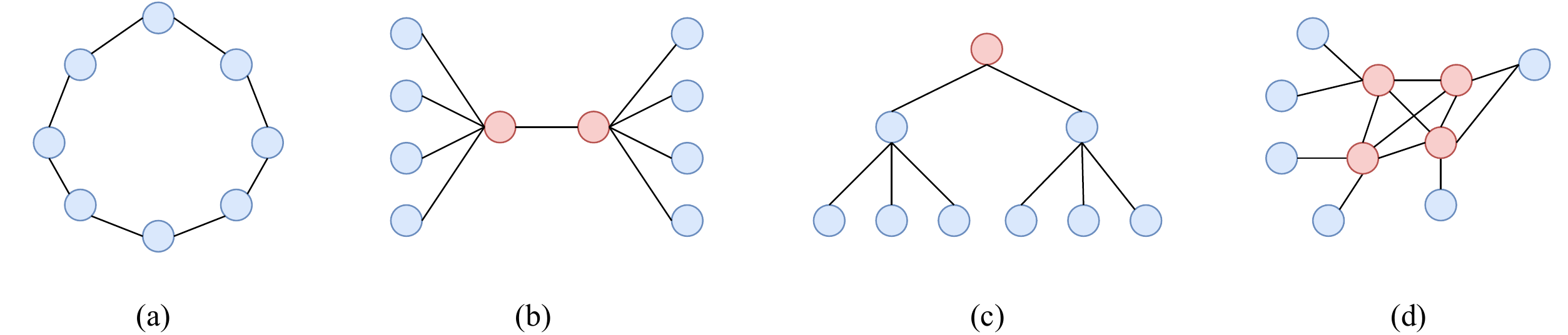}
	\caption{\rone{Four types of decentralized topologies. (a) Ring topology; (b) Star topology with two star agents; (c) Tree topology; (d) Partial mesh topology.}}\label{fig:topology}
\end{figure*}

\begin{table*}[t!]\centering
	\resizebox{\textwidth}{!}{
		\begin{tabular}{c@{}c@{}c@{}c}
			\includegraphics[width=0.33\textwidth]{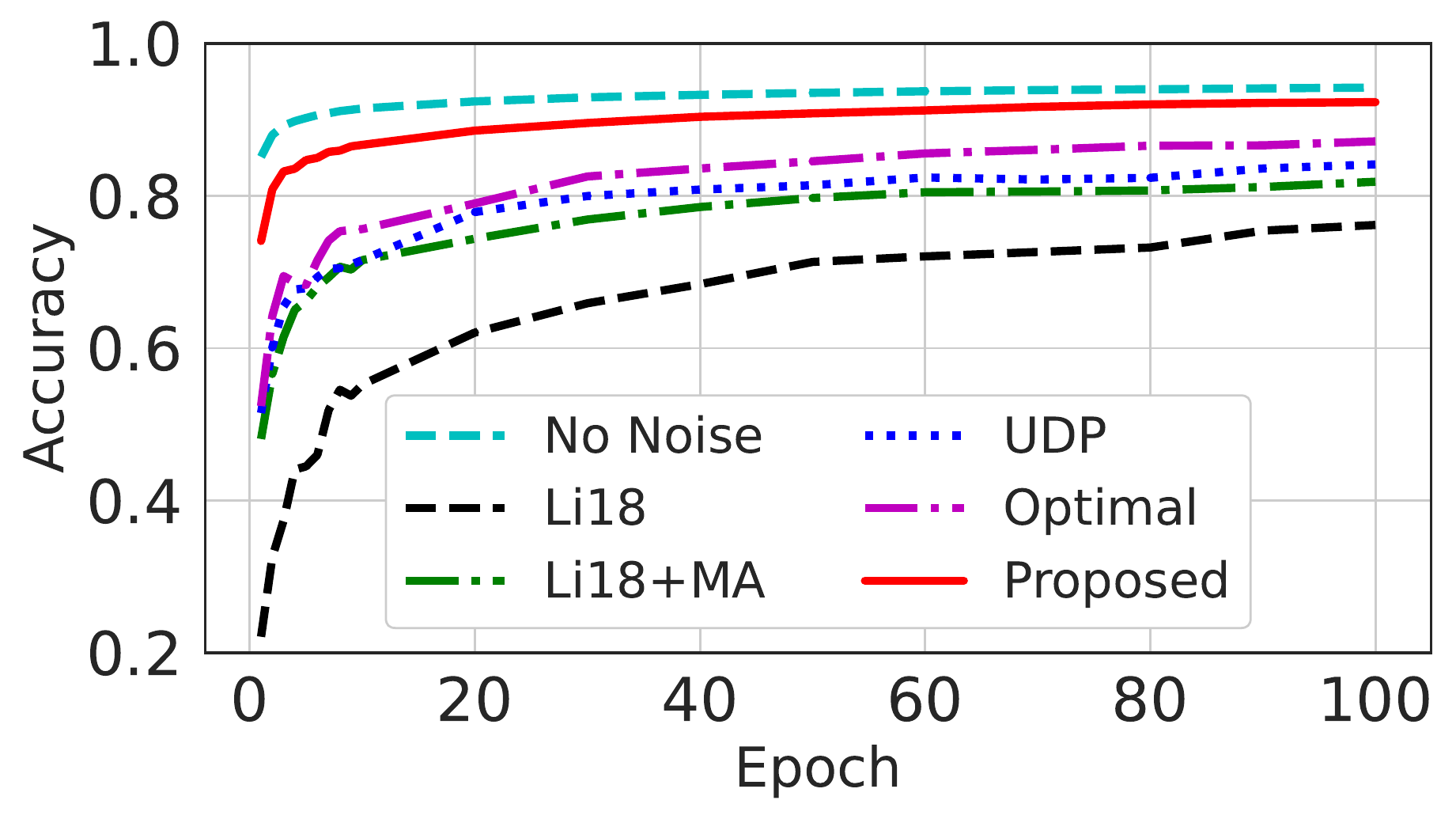}&\includegraphics[width=0.33\textwidth]{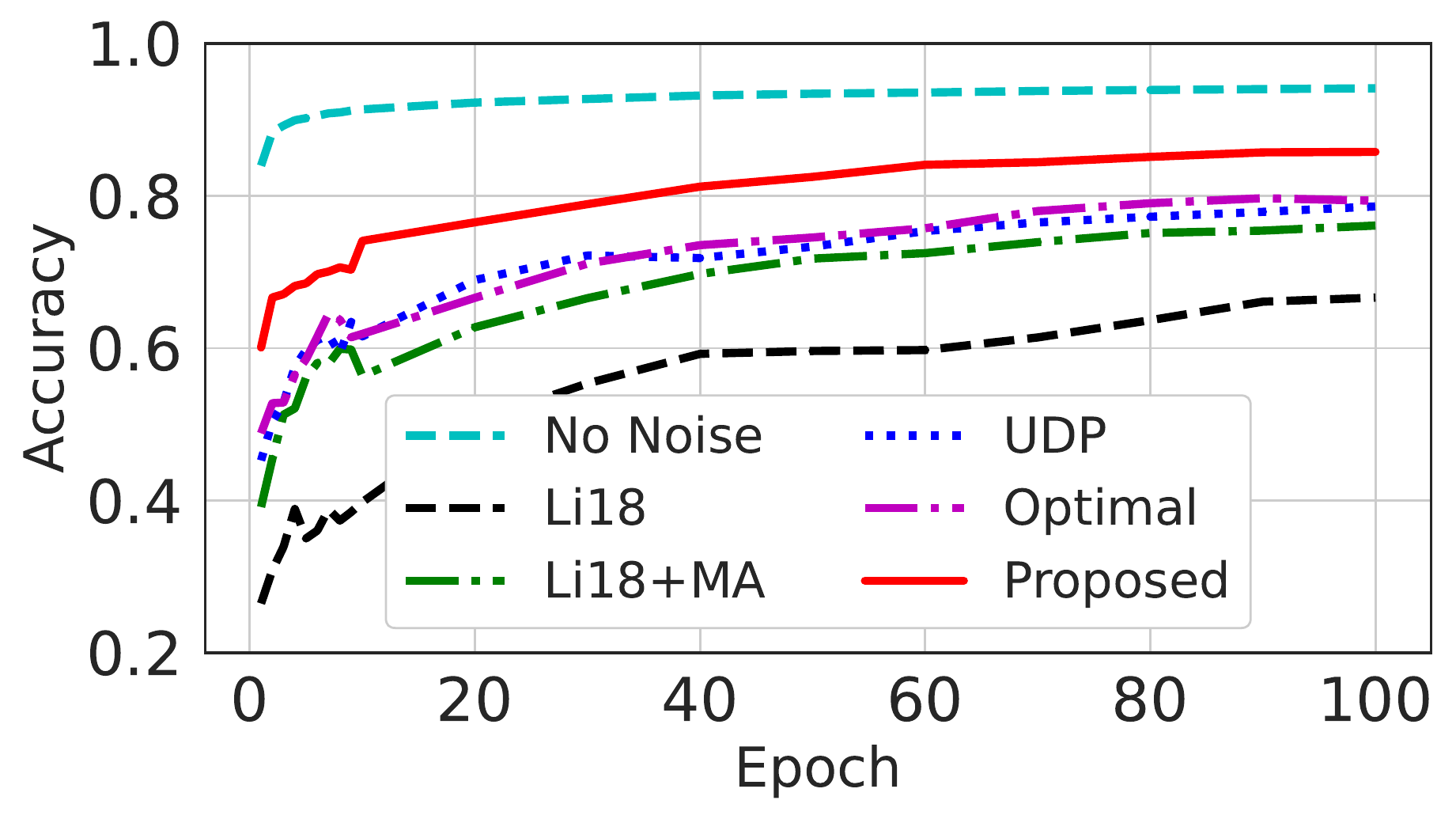}&\includegraphics[width=0.33\textwidth]{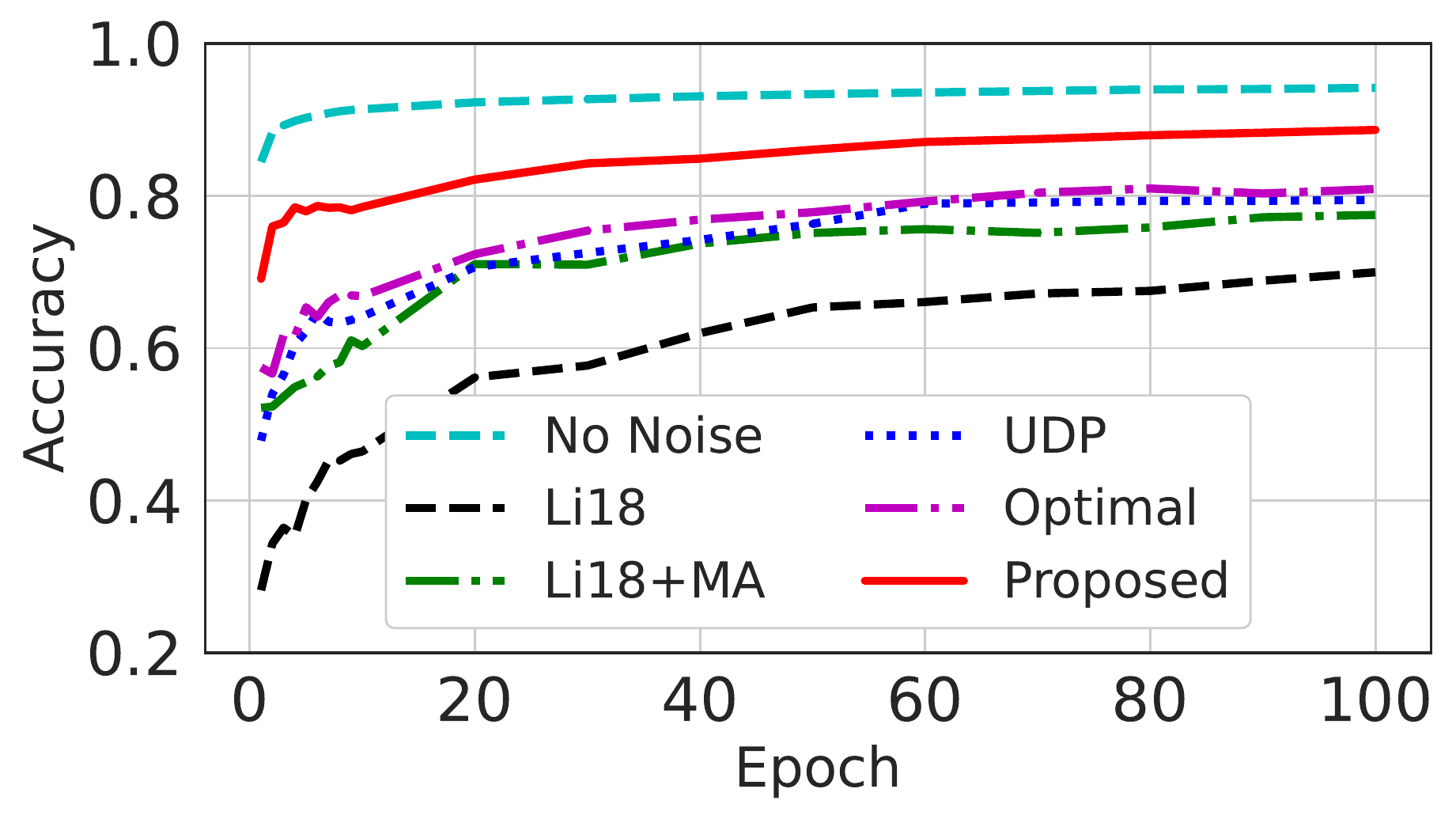}&\includegraphics[width=0.33\textwidth]{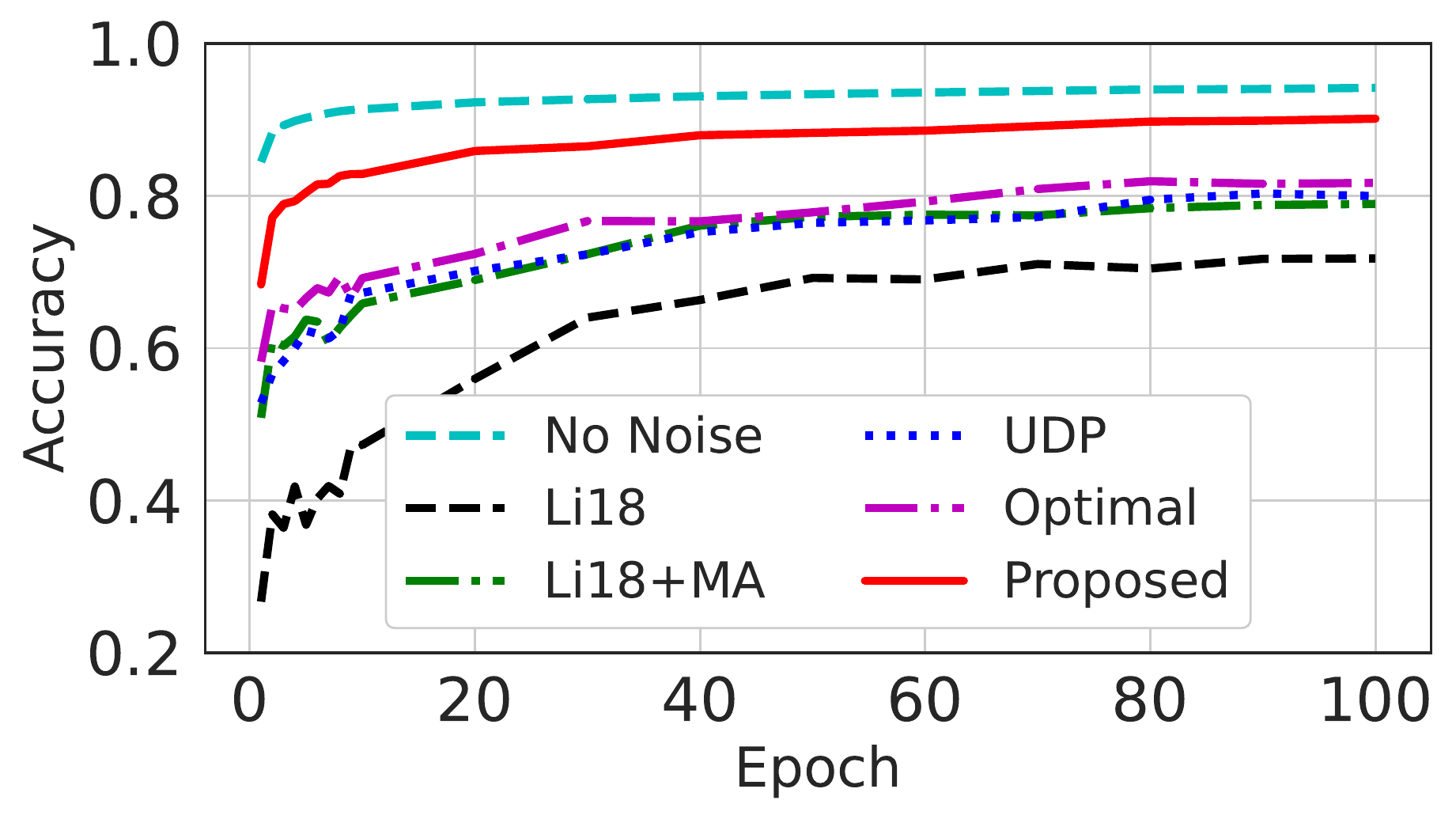}\\
			\includegraphics[width=0.33\textwidth]{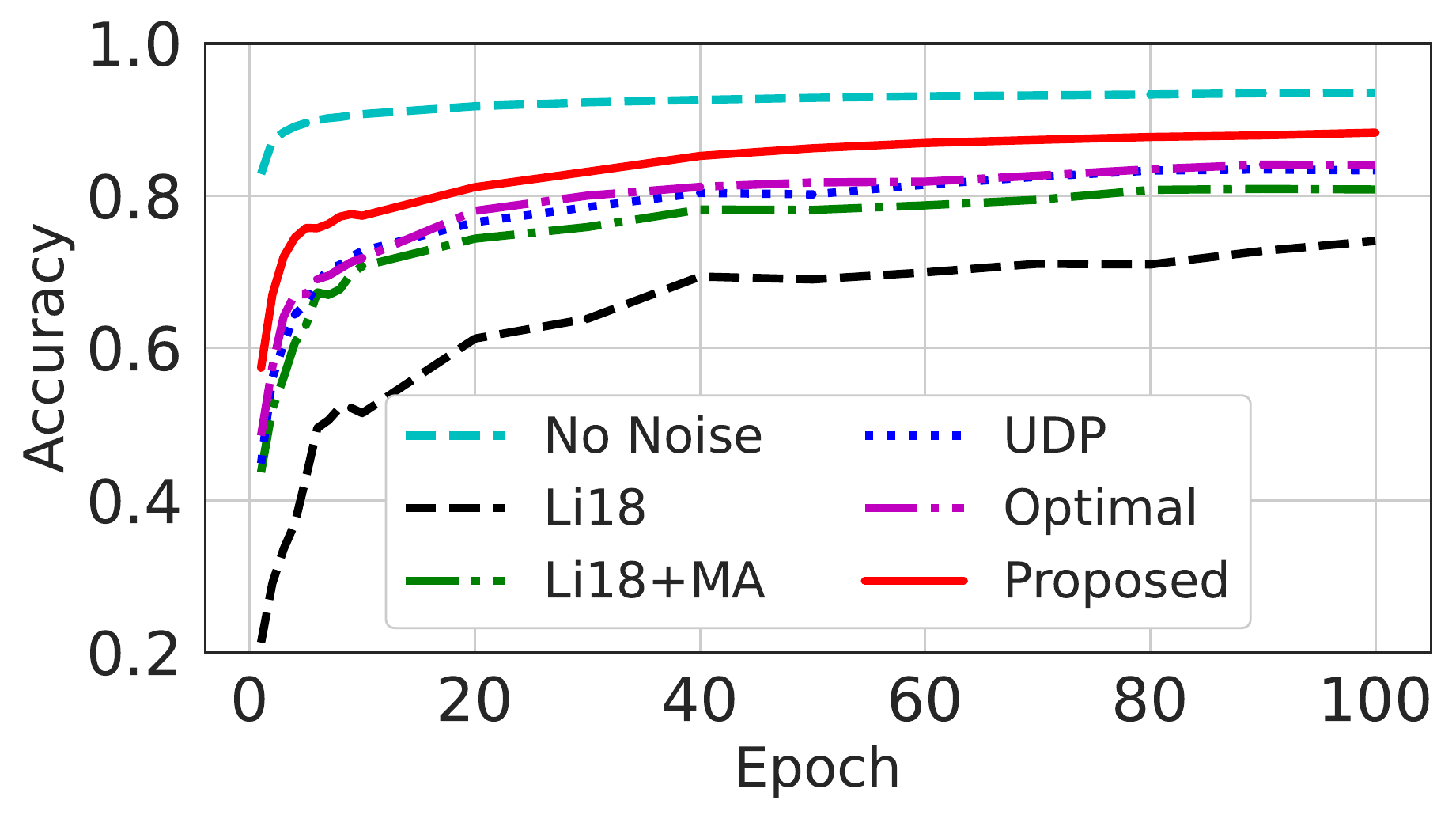}&\includegraphics[width=0.33\textwidth]{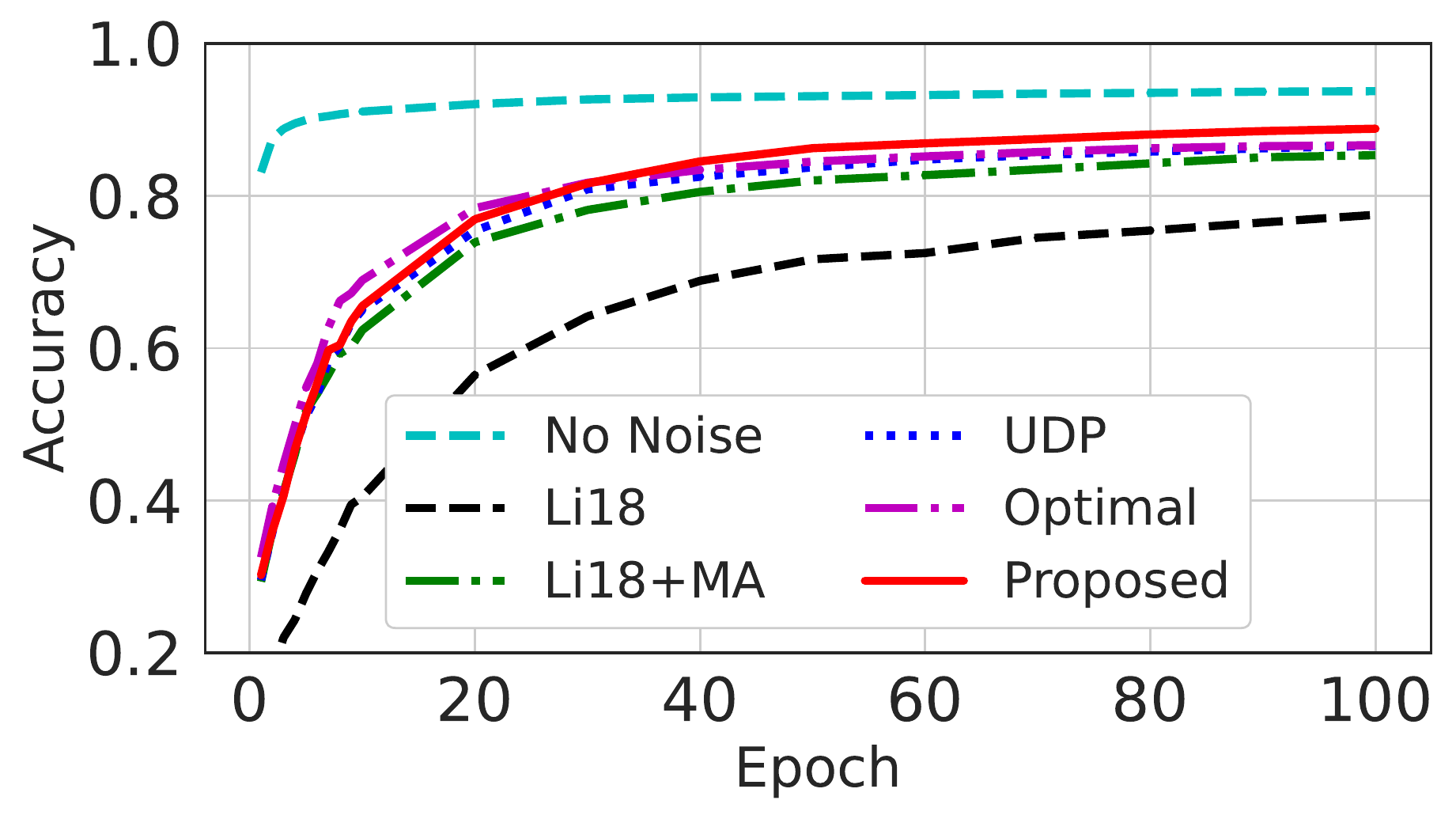}&\includegraphics[width=0.33\textwidth]{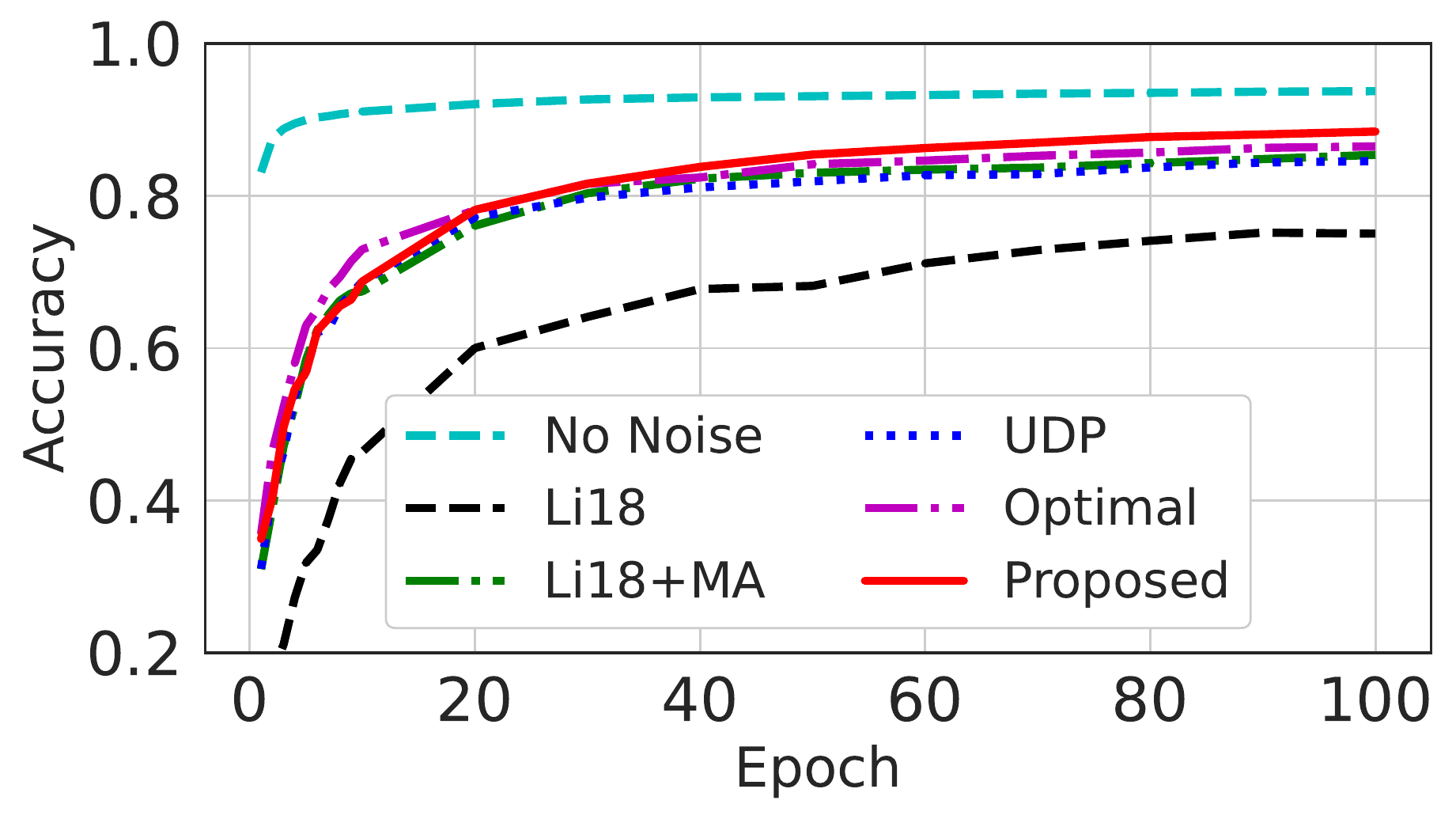}&\includegraphics[width=0.33\textwidth]{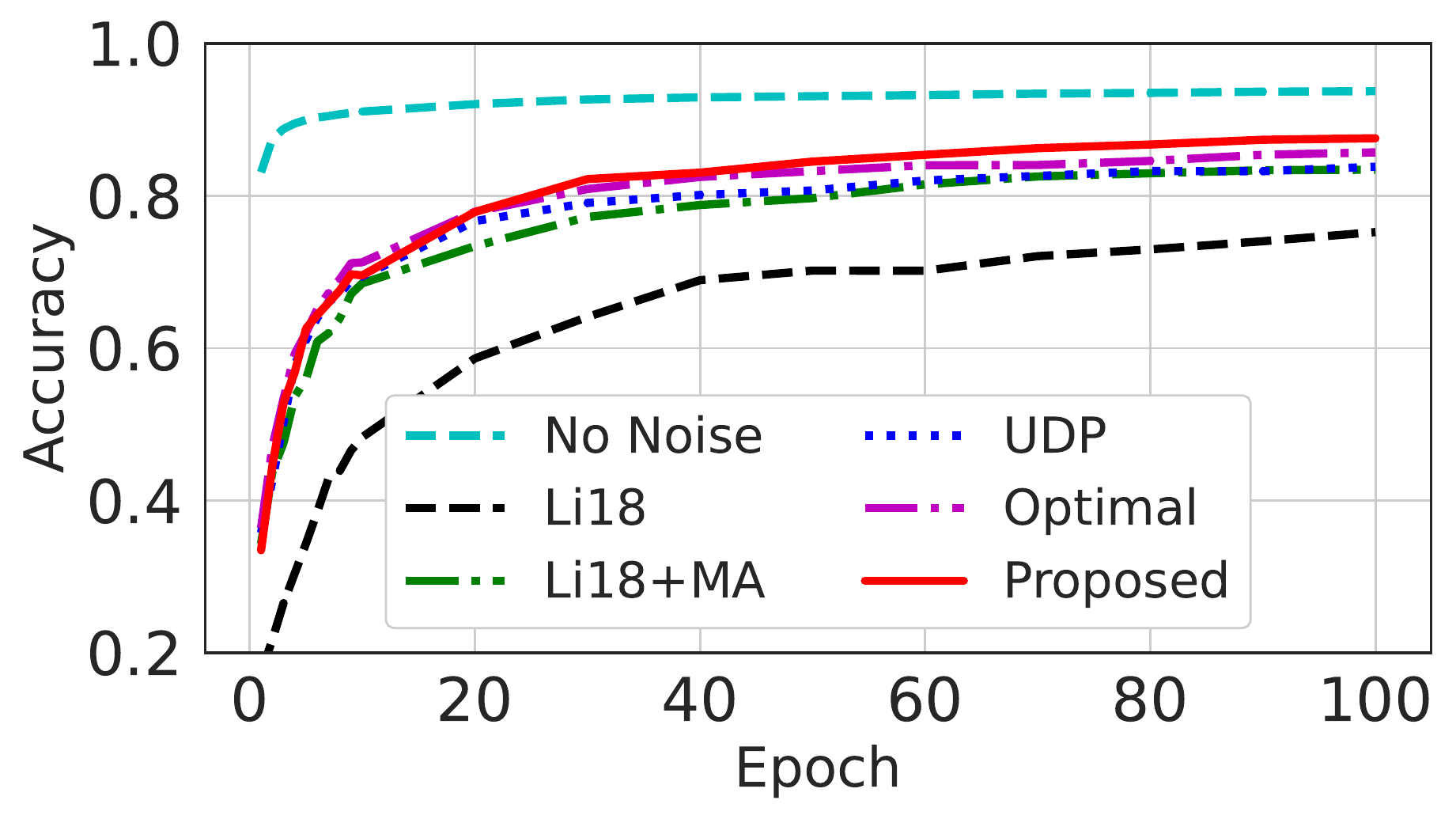}\\
			(a) Ring & (a) Star & (b) Tree & (c) Mesh
	\end{tabular}}
	\captionof{figure}{\rone{The average accuracy of the agents with different network topologies under both synchronous (first row) and asynchronous (second row) settings.}}\label{fig:topology_acc}
\end{table*}

\subsection{Effectiveness of Each Strategy}
Our DP-SGD learning protocols are composed of two strategies: topology-aware noise reduction (\textsc{Top}) and noise-aware noise decay (\textsc{ND}). We evaluate the integration of these two strategies in the above experiments. In this section, we measure the effectiveness of \textsc{Top} separately. Figures \ref{fig:sync_eff} and \ref{fig:async_eff} illustrate the performance comparison between \textsc{Top}, the integration \textsc{Top}+\textsc{ND}, and other DP-SGD algorithms.

We observe that in the synchronous mode, \textsc{Top} almost has the same performance as \textsc{Top}+\textsc{ND} at the first 20 epochs, as the reduced noise from \textsc{ND} strategy is quite small at the first two reduction steps (the noise is not reduced at the first reduction step). With more epochs, \textsc{Top}+\textsc{ND} is slightly better than \textsc{Top} only, caused by the effectiveness of \textsc{ND}. In the asynchronous mode, \textsc{Top} almost has the same performance as \textsc{Top}+\textsc{ND} especially when $\alpha$ equals 0.25.



\subsection{Results of a More Complicated Dataset}
\label{sec:cifar}
We also evaluate \AlgName on a more complicated training task over CIFAR10 dataset. The model to be trained is a Convolutional Neural Network, consisting of two max–pooling layers and three fully connected layers. The system settings and configurations are the same as the ones on MNIST. We set $\alpha$ and the connection rate as 0.25 and 0.2.

Figure \ref{fig:cifar} illustrates the experimental results in the synchronous and asynchronous modes. We observe that our solution (Proposed) outperforms prior DP-SGD algorithms and approaches the baseline (No Noise) as the training epoch increases in both of the two modes. The other four baselines even do not converge in the presence of Gaussian noise. The reason is that each parameter in the model needs to be appended with random noise to satisfy DP requirement. When the model becomes more complicated with more parameters, the overall divergence between the original model and the DP-protected model becomes larger, making it hard to converge. This scenario will never happen in our solution.

\rtwo{Our \AlgName is designed to be general for various learning tasks and datasets. In terms of the computational complexity, the protocols require each agent to calculate the scale of noise that is added to its estimates at each iteration. The cost of calculating the noise scale is a constant, while the calculation of adding noise to estimates is proportional to the number of parameters, which is negligible compared to the training overhead. So we believe our solution is practical and scalable to higher-dimensional datasets and more complex neural networks. As future work, we will evaluate \AlgName on larger-scale decentralized learning tasks. }

\begin{figure}[t!]
	\centering
	\subfigure[$\alpha = 0.5$]{\includegraphics[width=0.49\columnwidth]{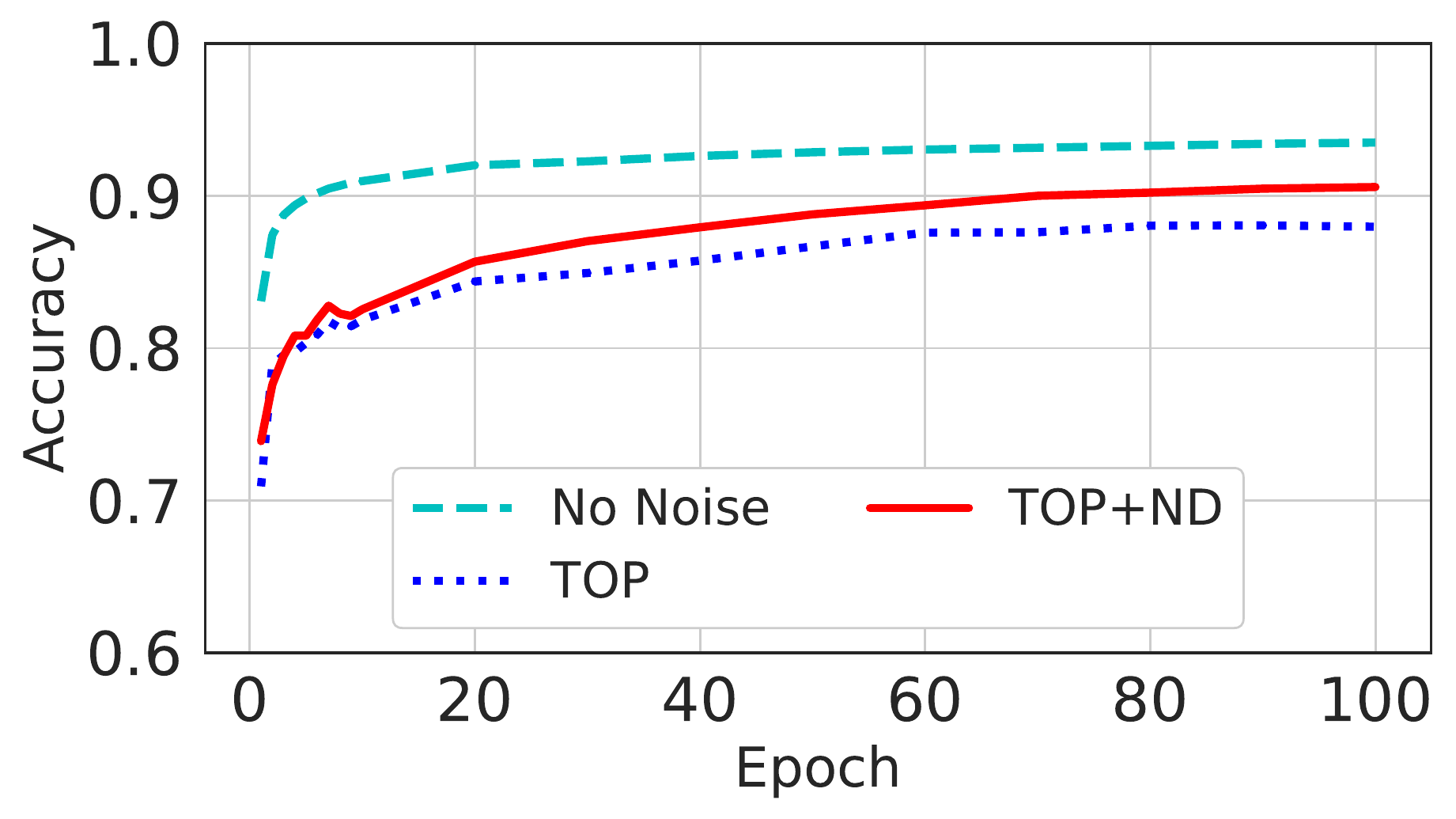}}
	\subfigure[$\alpha = 0.25$]{\includegraphics[width=0.49\columnwidth]{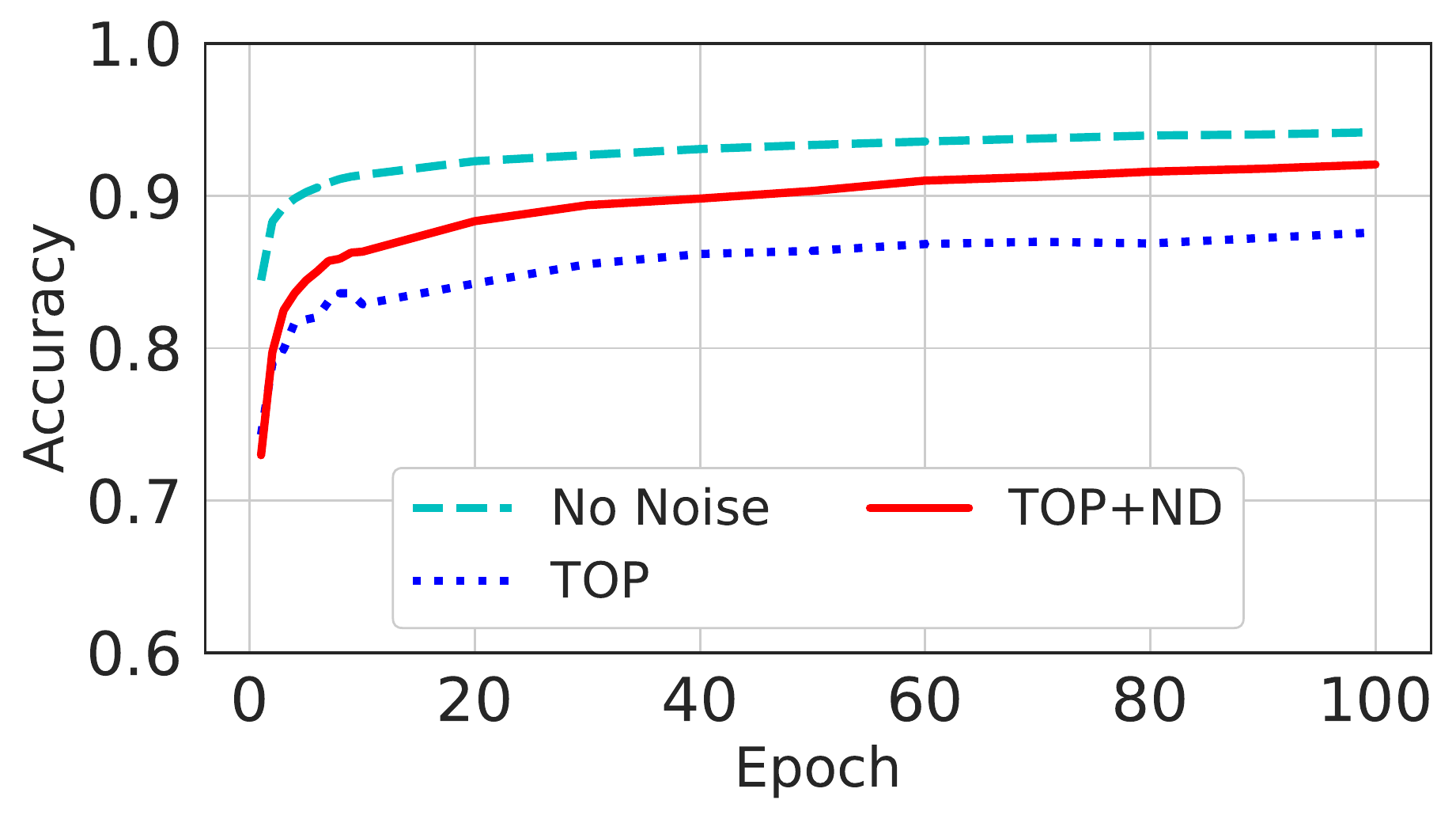}}
	\caption{The effectiveness of topology-aware noise reduction with different $\alpha$ values under synchronous settings.}
	\label{fig:sync_eff}
\end{figure}

\begin{figure}[t!]
	\centering
	\subfigure[$\alpha = 0.5$]{\includegraphics[width=0.49\columnwidth]{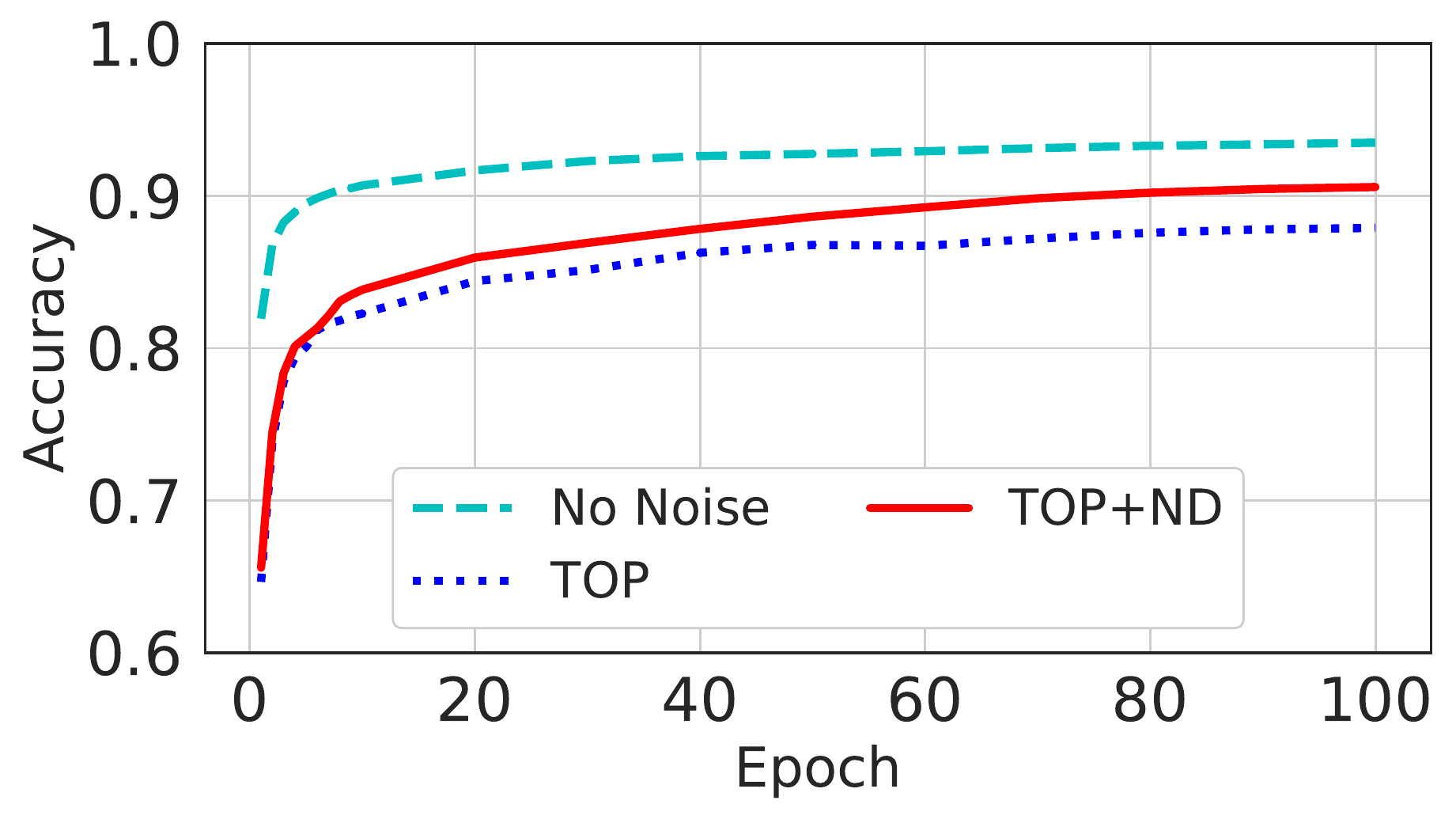}}
	\subfigure[$\alpha = 0.25$]{\includegraphics[width=0.49\columnwidth]{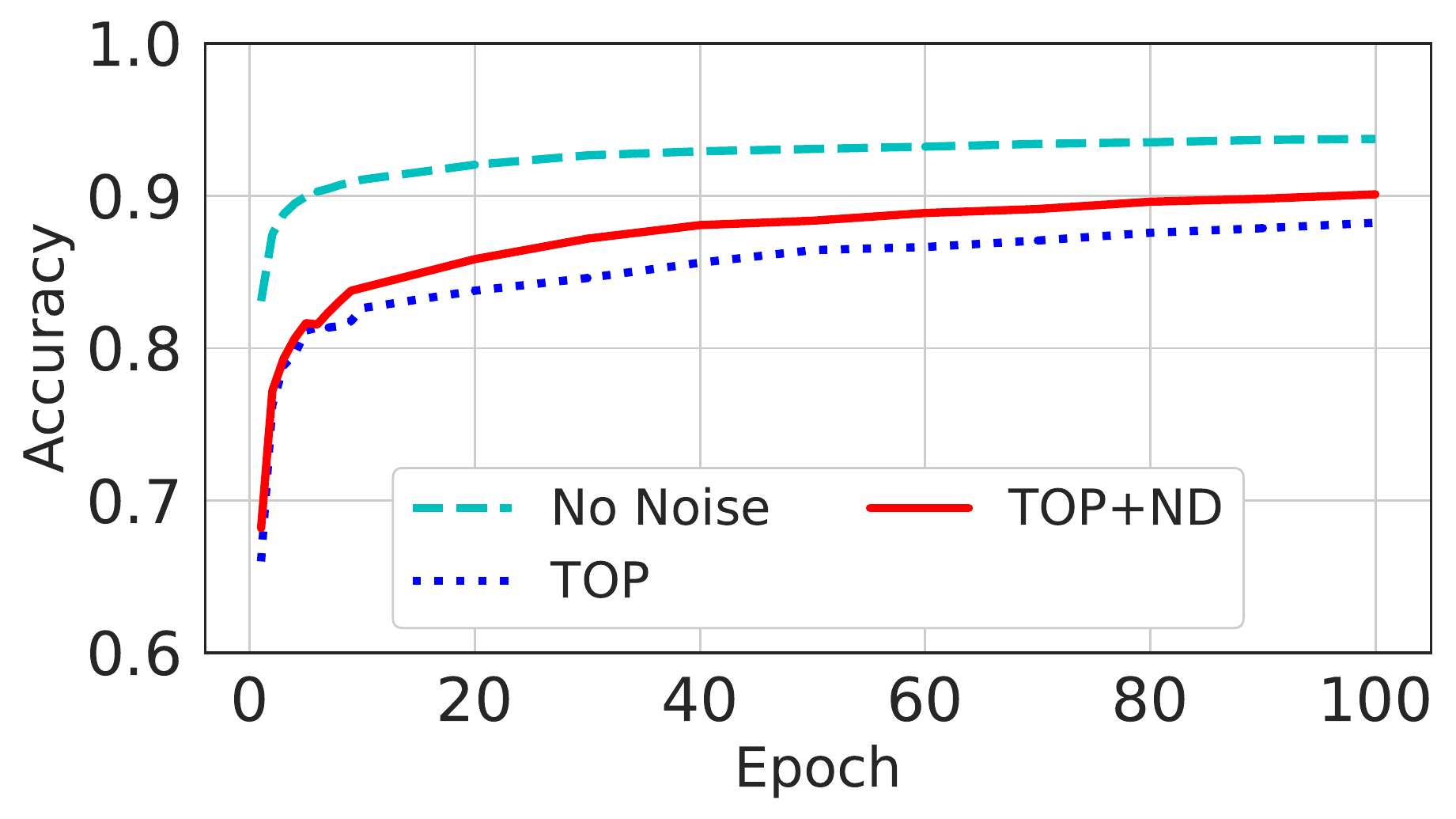}}
	\caption{The effectiveness of topology-aware noise reduction with different $\alpha$ values under asynchronous settings.}
	\label{fig:async_eff}
\end{figure}




\begin{figure}[t!]
	\centering
	\subfigure[Synchronous]{\includegraphics[width=0.49\columnwidth]{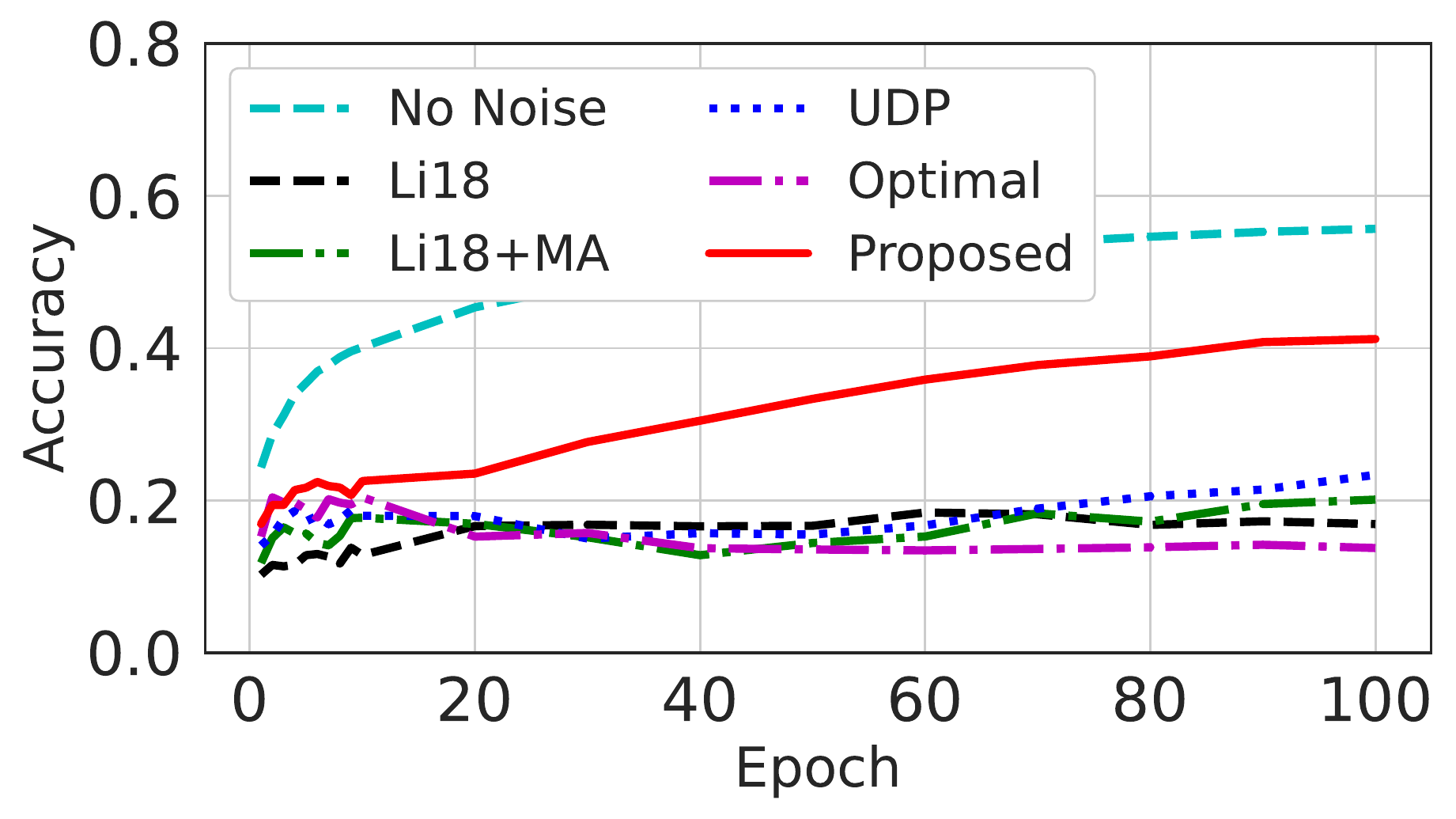}}
	\subfigure[Asynchronous]{\includegraphics[width=0.49\columnwidth]{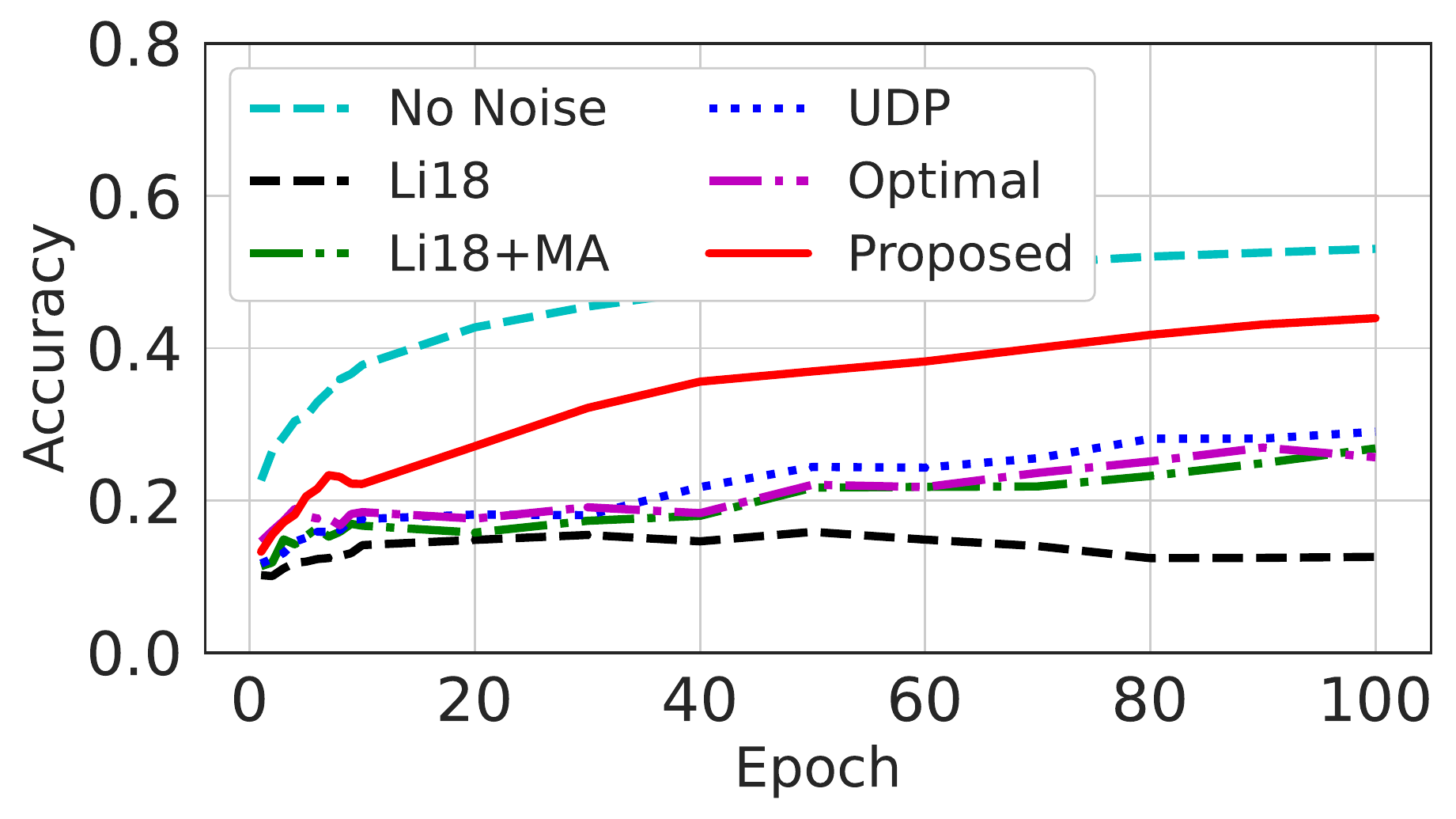}}
	\caption{\ronetwo{The average accuracy of the agents in different modes on CIFAR10.}}
	\label{fig:cifar}
\end{figure}

%% file: body/conclusion.tex
\section{Conclusion}
\label{sec:conclusion}
In this paper, we propose \AlgName, a novel DP-based method to preserve the privacy of decentralized learning systems. The topology-aware technique leverages the network topology to reduce the noise scale and improve model usability while still satisfying the DP requirement. We apply the time-aware noise decay technique to the decentralized systems to further optimize the model performance. We design learning protocols, which enables the topology-aware technique and adapts to both the synchronous and asynchronous learning modes. To the best of our knowledge, this is the first study to utilize network topology for DP optimization, and deploy DP protection to asynchronous decentralized systems. Formal analysis and empirical evaluations indicate that \AlgName can guarantee the privacy requirement, and achieve better trade-offs between privacy and usability under different system configurations.